\documentclass{article}
\usepackage{amsmath, amsthm, amssymb,color,fullpage,url,booktabs}  %cite
\usepackage[pdftex]{hyperref} %pagebackref
\usepackage{graphicx}

\newcommand{\nc}{\newcommand}
\nc{\rnc}{\renewcommand}

\newcommand{\bra}[1]{\left\langle #1\right|}
\newcommand{\ket}[1]{\left|#1\right\rangle}

\DeclareMathOperator{\conv}{conv}

\DeclareMathOperator{\inj}{inj}

\DeclareMathOperator{\poly}{poly}
\DeclareMathOperator{\polylog}{polylog}
\DeclareMathOperator{\tr}{tr}

\DeclareMathOperator{\Sep}{Sep}

\DeclareMathOperator{\NP}{\mathsf{NP}}

\def\SAT{\text{3-SAT}}

\def\be#1\ee{\begin{equation}#1\end{equation}}
\def\bea#1\eea{\begin{eqnarray}#1\end{eqnarray}}
\def\beas#1\eeas{\begin{eqnarray*}#1\end{eqnarray*}}
\def\ba#1\ea{\begin{align}#1\end{align}}
\def\bas#1\eas{\begin{align*}#1\end{align*}}
\def\bpm#1\epm{\begin{pmatrix}#1\end{pmatrix}}

\def\eq#1{(\ref{eq:#1})}

\def\L{\left} 
\def\R{\right}
\def\ra{\rightarrow}
\def\ot{\otimes}

\newtheorem{thm}{Theorem}
\newtheorem*{thm*}{Theorem}

\newtheorem{cor}[thm]{Corollary}
\newtheorem{lem}[thm]{Lemma}

\newtheorem{dfn}[thm]{Definition}
\newtheorem{proto}{Protocol}

\makeatletter
\newtheorem*{rep@theorem}{\rep@title}
\newcommand{\newreptheorem}[2]{%
\newenvironment{rep#1}[1]{%
 \def\rep@title{#2 \ref{##1} (restatement)}%
 \begin{rep@theorem}}%
 {\end{rep@theorem}}}
\makeatother

\newreptheorem{thm}{Theorem}
\newreptheorem{lem}{Lemma}

\def\eps{\varepsilon}

\def\cA{\mathcal{A}}
\def\cB{\mathcal{B}}
\def\cC{\mathcal{C}}
\def\cD{\mathcal{D}}

\def\cH{\mathcal{H}}

\def\cL{{\cal L}}

\def\cS{\mathcal{S}}

\def\bbC{\mathbb{C}}
\DeclareMathOperator*{\E}{\mathbb{E}}

\def\bbR{\mathbb{R}}

\def\benum{\begin{enumerate}}
\def\eenum{\end{enumerate}}
\def\bit{\begin{itemize}}
\def\eit{\end{itemize}}

\newcommand{\secref}[1]{Section~\ref{sec:#1}}

\newcommand{\lemref}[1]{Lemma~\ref{lem:#1}}
\newcommand{\thmref}[1]{Theorem~\ref{thm:#1}}

\newcommand{\algref}[1]{Algorithm~\ref{alg:#1}}

\nc{\todo}[1]{\textcolor{red}{todo: #1}}

\def\begsub#1#2\endsub{\begin{subequations}\label{eq:#1}\begin{align}#2\end{align}\end{subequations}}
\nc\qand{\qquad\text{and}\qquad}
\nc\mnb[1]{\medskip\noindent{\bf #1}}

\usepackage{boxedminipage}
\newtheorem{algorithm}[thm]{Algorithm}
\newenvironment{mybox}
{\center \noindent\begin{boxedminipage}{1.0\linewidth}}
{\end{boxedminipage}\noindent}

\begin{document}

%\title{Replacing Hierarchies by Nets \\ alt: 
\title{Estimating operator norms using covering nets}

\author{Fernando G.S.L. Brand\~ao \footnote{(a) Quantum Architectures and Computation Group, Microsoft Research, Redmond, WA and (b) Department of Computer Science, University College London WC1E 6BT. email: {\tt fbrandao@microsoft.com} }
\and Aram W. Harrow \footnote{
Center for Theoretical Physics, Massachusetts Institute of Technology.  email: {\tt aram@mit.edu}}}

\date{\vspace{-1cm}}
\maketitle 

\begin{abstract}
  We present several polynomial- and quasipolynomial-time
  approximation schemes for a large class of generalized operator
  norms. Special cases include the $2\ra q$ norm of matrices for
  $q>2$, the support function of the set of separable quantum states, 
  finding the least noisy output of entanglement-breaking quantum
  channels, and approximating the injective tensor norm for a map 
  between two Banach spaces whose
  factorization norm through $\ell_1^n$ is bounded.

  These reproduce and in some cases improve upon the performance of
  previous algorithms by Brand\~ao-Christandl-Yard~\cite{BCY10} and
  followup work, which were based on the Sum-of-Squares hierarchy and whose analysis used techniques from quantum information such as the monogamy principle of entanglement.
  Our algorithms, by contrast, are based on brute force enumeration
  over carefully chosen covering nets.  These have the advantage of
  using less memory, having much simpler proofs and giving new
  geometric insights into the problem.   Net-based algorithms for similar problems were also presented by Shi-Wu~\cite{ShiW12} and Barak-Kelner-Steurer~\cite{BKS13}, but in each case with a run-time that is exponential in the rank of some matrix.  We achieve polynomial or quasipolynomial runtimes by using the much smaller nets that exist in $\ell_1$ spaces.  This principle has been used in learning theory, where it is known as Maurey's empirical method.
  
  %They have the disadvantage of
  %requiring the input to be given in a more structured format.

\iffalse
estimating norms of superoperators and 
  We give a quasipolynomial-time algorithm based on enumerating over
  nets for optimizing a linear function over the set of separable
  quantum states, for the class of functions associated to measurements
  implemented by local operations and one-way classical communication
  (1-LOCC). The algorithm matches the performance of the 2010 paper of
  Brand\~ao, Christandl and Yard \cite{BCY10}, which used quantum
  information-theoretic methods to show that the Sum-of-Squares
  hierarchy solves the problem in quasipolynomial time. We then show
  that the net enumeration approach allows us to improve on
  \cite{BCY10} and its follow-ups in two directions. First we give
  faster algorithms for special cases of the problem, including a
  polynomial-time approximating scheme (PTAS) for 1-LOCC measurements
  for which Bob's measurements are low rank. Second we give a similar
  algorithm for a generalization of the problem: estimating the
  injective tensor norm for a map between two Banach spaces, whose
  factorization norm through $\ell_1^n$ is bounded. As a particular
  case we find the first PTAS for estimating the maximum output
  Schatten-$\alpha$ norms of entanglement-breaking quantum channels
  for any $\alpha \geq 1$, and for estimating the hypercontractive $2
  \rightarrow s$ norm of a matrix for any $s \geq 2$.
\fi
\end{abstract}

\section{Introduction}

Given a $n \times m$ matrix $M$, its operator norm is given by $\Vert M \Vert = \max _{x \in \mathbb{C}^m} \Vert M x \Vert_2 / \Vert x \Vert_2$, with $\Vert x \Vert_2 = (\sum_i |x_i|^2)^{\frac{1}{2}}$ the Euclidean norm. The operator norm is also given by the square root of the largest eigenvalue of $M^{\cal y}M$ and thus can be efficiently computed. There are numerous ways of generalizing the operator norm, e.g. by considering tensors instead of matrices, by changing the Euclidean norm to another norm, or by considering other vector spaces instead of $\mathbb{C}^m$.  Although such generalizations are very useful in applications, they can be substantially harder to compute than the basic operator norm, and in many cases we still do not have a good grasp of the computational complexity of computing, or even only approximating, them.  In some cases quasipolynomial algorithms are known, usually based on semidefinite programming (SDP) hierarchies, and in other cases quasipolynomial hardness results are known.  These are partially overlapping so that some problems have sharp bounds on their complexity and for others there are exponential gaps between the best upper and lower bounds.   As we will discuss below, the complexity of these problems is not only a basic question in the theory of algorithms, but also is closely related to the unique games conjecture and the power of multiprover quantum proof systems.

In this paper we give new algorithms for several variants of the basic operator norm of interest in quantum information theory, theoretical computer science, and the theory of Banach spaces.  Unlike most past work which was based on SDP hierarchies, our algorithms simply enumerate over a carefully chosen net of points.  This yields run-times that often match the SDP hierarchies and sometimes improve upon them.  Besides improved performance, our algorithms have the advantage of being based on simple geometric properties of spaces we are optimizing over, which may help explain which types of norms are amenable to quasipolynomial optimization.
 In particular we consider the following four optimization problems in this work:

\vspace{0.2 cm}

\textbf{Optimization over Separable States:} An important problem in quantum information theory is to optimize a linear function over the set of separable (i.e. non-entangled) states, defined as bipartite density matrices that can be written as a convex combination of tensor product states. This problem is closely related to the task of determining if a given quantum state is entangled or not (called the quantum separability problem) and to the computation of several other quantities of interest in quantum information, including the optimal acceptance probability of quantum Merlin-Arthur games with unentangled proofs, optimal entanglement witnesses, mean-field ground-state energies, and measures of entanglement; see \cite{HM13} for a review of many of these connections.

Given an operator $M$ acting on the bipartite vector space $\mathbb{C}^{d_1} \otimes \mathbb{C}^{d_2}$ the support function of $M$ on the set of separable states is given by
\begin{equation} \label{basicproblem}
h_{\text{Sep}(d_1, d_2)}(M) := \max_{\alpha \in \cD_{d_1}, \beta \in \cD_{d_2}} \tr[M (\alpha \otimes \beta)],
\end{equation}
with $\cD_{d}$ the set of density matrices on $\mathbb{C}^{d}$ ($d \times d$ positive semidefinite matrices of unit trace). Our goal is to approximate $h_{\Sep(d_1, d_2)}(M)$.  For $M\in L(\bbC^{d^n})$, define
\be  h_{\Sep^n(d)}(M) = \max_{\alpha_1,\ldots,\alpha_n \in \cD_d} \tr[M(\alpha_1 \ot \cdots \ot \alpha_n)].\ee

The first result on the complexity of computing $h_{\Sep(d_1, d_2)}$
was negative: Gurvits showed that the problem is NP-hard for
sufficiently small additive error (inverse polynomial in $d_1d_2$)
\cite{Gurvits03}. Then \cite{HM13} showed there is no $\exp(O(\log^{2
  - \Omega(1)}(d_1d_2)))$ time algorithm even for a constant error
additive approximation of the quantity, assuming the exponential time
hypothesis (ETH\footnote{The ETH is the conjecture that 3-SAT instances of length $n$ require time $2^{\Omega(n)}$ to solve.  This is a plausible conjecture for deterministic, randomized or quantum computation, and each version yields a corresponding lower bound on the complexity of estimating $h_{\Sep}$.}). This left open the question whether there are
quasipolynomial-time algorithms (i.e. of time $\exp(\polylog(d_1, d_2)))$.

In \cite{BCY10} it was shown that this is indeed the case at least for
a class of linear functions: namely those corresponding to quantum
measurements that can be implemented by local operations and
one-directional classical communication (one-way LOCC or 1-LOCC). For this particular class of measurements the problem can be solved with error $\delta$ in time $\exp\left( O \left( \delta^{-2} \log(d_1)\log(d_2) \right)  \right)$. The proof was based on showing that the hierarchy of semidefinite programs for the problem introduced in 2004 by Doherty, Parrilo and Spedalieri \cite{DohertyPS04} (which is an application of the more general Sum-of-Squares (SoS) hierarchy, also known as the Lasserre hierarchy, to the separability problem) converges quickly. The approach of \cite{BCY10} was to use ideas from quantum information theory (monogamy of entanglement, entanglement measures, hypothesis testing, etc) to find good bounds on the quality of the SoS hierarchy. Since then several follow-up work gave different proofs of the result, but always using quantum information-theoretic ideas \cite{BH-local, LW14, BC11, Yang06}. 

A corollary of \cite{BCY10} and the other results on 1-LOCC $M$ is
that $h_{\Sep(d_1, d_2)}(M)$ can also be approximated for a different
class of operators $M$: those with small Hilbert-Schmidt norm $\Vert M
\Vert_{\text{HS}} := \tr(M^{\cal y}M)^{\frac{1}{2}}$. Ref.~\cite{BCY10b} showed that also in this case there is a quasipolynomial-time algorithm 
for estimating Eq. (\ref{basicproblem}). An interesting subsequent
development was the work of Shi and Wu \cite{ShiW12} (see also
\cite{BKS13}), who gave a different algorithm for the problem based on
enumerating over nets. It was left as an open question whether a
similar approach could be given for the case of one-way LOCC
measurements (which is more relevant both physically \footnote{The
  one-way LOCC norm gives the optimal distinguishably of two
  multipartite  quantum states when only local measurements can be
  done, and the parties can coordinate by one-directional
  communication. See \cite{MWW09} for a discussion of its power.} and in terms of applications; see again
\cite{HM13}).   

\vspace{0.2 cm}

\textbf{Estimating the Output Purity of Quantum Channels:} Another
important optimization problem in quantum information theory consists of determining how much noise a quantum channel introduces. A quantum channel models a general physical evolution and is given mathematically by a completely positive trace preserving map $\Lambda : \cD_{d_1} \rightarrow \cD_{d_2}$. One way to measure the level of noise of the channel is to compute the maximum over states of the output Schatten-$\alpha$ norm, for a given $\alpha > 1$:
\begin{equation}
\Vert  \Lambda \Vert_{1\ra \alpha} := 
\max_{\rho \in \cD_{d_1} }\Vert \Lambda(\rho) \Vert_{\alpha},
\end{equation}
with $\Vert Z \Vert_{\alpha} =
\tr(|Z|^{\alpha})^{\frac{1}{\alpha}}$. The quantity $\Vert  \Lambda
\Vert_{1\ra \alpha}$ varies from one, for an ideal channel, to
$d_2^{-1+\alpha^{-1}}$ for the depolarizing channel mapping all states
to the maximally mixed state. This optimization problem has been
extensively studied, in particular because for $\alpha\approx 1$ it is related to the 
Holevo capacity of the channel, whose regularization
gives the classical capacity of the channel (i.e. how many reliable
bits can be transmitted per use of the channel). 

It was shown in \cite{HM13} that, assuming ETH, there is no
algorithm that runs in time 
$\exp( O( \log^{2 - \Omega(1)}{(d)}) )$ 
 %$\exp \left( O \left( \log^{2 - \Omega(1)}{d} \right) \right )$ 
 and can decide if $\Vert \Lambda
\Vert_{1\ra\alpha}$ is one or smaller than $\delta$ (for any fixed
$\delta > 0$ and $\alpha > 1$) for a general quantum channel $\Lambda
: {\cal D}_{d} \rightarrow {\cal D}_{d}$. On the algorithmic side,
nothing better than exhaustive search over the input space (taking
time $\exp(\Omega(d_1))$) is known. 

An interesting subclass of quantum channels, lying somewhere between classical channels and fully quantum channels, are the so-called entanglement-breaking channels, which are the channels that cannot be used to distribute entanglement. Any entanglement-breaking quantum channel $\Lambda$
% : {\cal D}_{d_1} \rightarrow {\cal D}_{d_2}$ 
can be written as \cite{HSR04}:
\begin{equation}
\Lambda(\rho) := \sum_i \tr(X_i \rho) Y_i, 
\end{equation}
with $Y_i \geq 0$, $\tr(Y_i) = 1$ quantum states and $X_i \geq 0$, and
$\sum_i X_i = I$ a quantum measurement. Because of their simpler form,
one can expect that there are more efficient algorithms for computing
the maximum output norm of entanglement-breaking channels. However
until now no algorithm better than exhaustive search was known either
(apart from the case $\alpha=\infty$ where the Sum-of-Squares hierarchy can be used and analyzed using \cite{BCY10}).

\vspace{0.2 cm}

\textbf{Computing $p \rightarrow q$ Norms:} Given a $d_1 \times d_2$ matrix $A$ we define its $p \rightarrow q$ norm by
\begin{equation}
\Vert A \Vert_{p \rightarrow q} := \max_{x \in \mathbb{C}^{d_2}} \frac{\Vert A x \Vert_{q}}{\Vert x \Vert_p},
\qquad
\|x\|_p := \L(\sum_{i=1}^{d_2} |x|^p\R)^{1/p} 
\end{equation}

Such norms have many different applications, such as in
hypercontractive inequalities and determining if a graph is a
small-set expander \cite{BHKSZ12}, to oblivious
routing~\cite{BhaskaraV11} and robust optimization~\cite{Steinberg05}. However we do not have a complete understanding of the complexity of computing them. For $2 < q \leq p$ or $q \leq p < 2$, it is $\NP$-hard to approximate them to any constant factor \cite{BhaskaraV11}. In the regime $q > p$ (the one relevant for hypercontractivity and small-set expansion) the only known hardness result is that to obtain any (multiplicative) constant-factor approximation for the $2 \rightarrow 4$ norm of a $n \times n$ matrix is as hard as solving $\SAT$ with $\tilde O(\log^2(n))$ variables \cite{BHKSZ12}. 
%In the appendix we extend this result and show that to obtain any constant-factor approximation for the $2 \rightarrow q$ norm of a $n \times n^{O(q)}$ matrix is as hard as solving $\SAT$ with $O(\log^2(n))$ variables. 

On the algorithmic side, besides the $2 \rightarrow 2$ and $2 \rightarrow \infty$ norms being exactly computable in polynomial time, Ref. \cite{BHKSZ12} showed that one can use the Sum-of-Squares hierarchy to compute in time $\exp(O(\log^{2}(n) \varepsilon^{-2}))$ a number $X$ s.t. 
\begin{equation} \label{baraketalbound}
\Vert A \Vert^4_{2 \rightarrow 4} \leq X \leq \Vert A \Vert^4_{2 \rightarrow 4}  + \varepsilon \Vert A \Vert_{2 \rightarrow 2}^2 \Vert A \Vert_{2 \rightarrow \infty}^2.
\end{equation}

Whether similar approximations can be obtained for $2 \rightarrow q$ norms for other values of $q$ was left as an open problem.

\vspace{0.2 cm}

\textbf{Computing the Operator Norm between Banach Spaces:}  These
problems are all special cases of the following general question.
Given a map $T:\cA\ra \cB$ between Banach spaces $\cA, \cB$, can we
approximately compute the following operator norm?
\be \|T\|_{\cA \ra \cB} := \sup_{x\neq 0} \frac{\|Tx\|_\cB}{\|x\|_{\cA}}\ee

\subsection{Summary of Results}

In this paper we give new algorithmic results for the four problems
discussed above. 
They can be summarized as follows.

\vspace{0.2 cm}

\textbf{Separable-state optimization by covering nets:} We give a
different algorithm for optimizing linear functions over separable
states (corresponding to one-way LOCC measurements) based on
enumerating over covering nets (see Algorithm \ref{alg:basic}). The
complexity of the algorithm matches the time complexity of
\cite{BCY10} (see Theorem \ref{basicthm}). The proof does not use
information theory in any way, nor the SoS hierarchy. Instead the main
technical tool is a matrix version of the Hoeffding bound (see Lemma
\ref{lem:Hoeff}). It gives new geometric insight into the problem and
gives arguably the simplest and most self-contained proof of the result
to date. It also gives an explicit rounding (as does \cite{BH-local}
but in contrast to \cite{BCY10, LW14, BC11, Yang06}).  

For particular subclasses of one-way LOCC measurements our algorithm
improves the run time of \cite{BCY10}. One example is the case where
Bob's measurement outcomes are low rank, in which we find a $\poly(d_2) d_1^{O(\varepsilon^{-2})}$-time algorithm.

%Algorithms based on the $k$-extendable hierarchy implied that $\QMA_m^{\text{1-LOCC}}(2)_{c,s} \subseteq
%\QMA_{O(m^2/\eps^2)}(1)_{c,s+\eps}$.  Our net-based algorithms yield a different and incomparable upper bound on
%$\QMA_m^{\text{1-LOCC}}(2)_{c,s}$.   Let $\beta_k\TIME(T)$ be the class of languages decidable using $k$ nondeterministic bits and time
%$T$. Our algorithm shows that $\QMA_m^{\text{1-LOCC}}(2)_{c,s} \subseteq \beta_{O(m^2/(c-s)^2)}\BP\TIME(2^{O(m)})$.  

%{\em update:} In fact, using the Shi-Wu observation about multiplicative weights, this can be
%done using $O(m^2)$ nondeterministic bits, $O(m)$ space and
%$\exp(O(m))$ time.

%We note one disadvantage of our approach is that we need an explicit
%decomposition of the measurement $M$ as a one-way LOCC protocol.  See
%\secref{explicit} for more discussion of this point.

\vspace{0.2 cm}

\textbf{Generalization to arbitrary operator norms:} Computing
$h_{\Sep}$ is mathematically equivalent to computing the $1\ra \infty$
norm of a quantum channel, or more precisely the $S_1\ra S_\infty$
norm where $S_\alpha$ denotes the Schatten-$\alpha$ norm.  This
perspective will help us generalize the scope of our algorithm, to
estimating the $S_1\ra \cB$ norm for a general Banach space $\cB$.
The analysis of this algorithm is based on tools from asymptotic
geometric analysis, and we will see that its efficiency depends on
properties of $\cB$ known as the Rademacher type and the modulus of
uniform smoothness.  Besides generalizing the scope of the algorithm,
this also gives more of a geometric explanation of its performance. We
focus on two special cases of the problem:  

\benum

\item {\bf maximum output norm:} A particular case of the generalization is the problem of computing the maximum output purity of a quantum entanglement-breaking channel (measured in the Schatten-$\alpha$ norms). We prove that for any $\alpha > 1$ one can compute $\Vert  \Lambda \Vert_{1\ra \alpha}$  in time $\poly(d_2) d_1^{O(\varepsilon^{-2})}$ to within additive error $\varepsilon$. (see Corollary \ref{cor:maxoutputentropy}). In contrast known hardness results \cite{HM10, HM13} show that no such algorithm exists for general quantum channels (under the exponential time hypothesis).   Previously the entanglement-breaking case was not known to be easier.

\item \textbf{matrix $2 \rightarrow q$ norms:}  As a second particular case of the general framework we extend the approximation of \cite{BHKSZ12} to the $2 \rightarrow 4$ norm, given in Eq. (\ref{baraketalbound}), to the $2 \rightarrow q$ norms for all $q \geq 2$ (see Corollary \ref{cor:2to4}).  

\eenum

\vspace{0.2 cm}

\textbf{Operator norms between Banach spaces:} This framework can be further generalized to estimating the operator norm of any linear map from $\cA \ra \cB$ for Banach spaces $\cA, \cB$.  Here we have replaced $S_1$ with any finite-dimensional Banach space $\cA$ whose norm can be computed efficiently.  When applied to an operator $\Lambda$, the approximation error scales with the $\cA \ra \ell_1^n \ra \cB$ {\em factorization norm}, which is the minimum of $\|\Lambda_1\|_{\ell_1^n\ra B}\|\Lambda_2\|_{\cA\ra \ell_1^n} $ such that 
$\Lambda = \Lambda_1\Lambda_2$.    Factorization norms have applications to communication complexity~\cite{LinialS07,LinialS09}, Banach space theory~\cite{Pietsch07}, and machine learning~\cite{LeeRRST10}, and here we argue that they help explain what makes the class of 1-LOCC measurements uniquely tractable for algorithms.  In \secref{inj-norm} we describe an algorithm for this general norm estimation problem, which to our knowledge previously had no efficient algorithms.  This problem equivalently can be viewed as computing the injective tensor norm of two Banach spaces.

We remark that this generalization is not completely for free, so we cannot simply derive all our other algorithms from this final one.  In the case where $\cA=S_1^d$ (which corresponds to all of our specific applications), we are able to easily sparsify the input; i.e. given $\sum_{i=1}^n A_i \ot B_i$, we can reduce $n$ to be $\poly(d)$ without loss of generality.  For general input spaces $\cA$ we do not know if this is possible.  Also, the case of $h_{\Sep}$ is much simpler, and so it may be helpful to read it first.

\subsection{Comparison with prior work}
\label{sec:net-compare}
As discussed in the introduction, previous algorithms for separable-state optimization (and as a corollary, the $2\ra 4$ norm) have been obtained using SDP hierarchies.  Our algorithms generally match or improve upon their parameters, but with the added requirement for the separable-state problem that the input be presented in a more structured form.

Several parallels between LP/SDP hierarchies and net-based algorithms have been developed for other problems.  The first example of this was Ref.~\cite{KlerkLP06} which gave both types of algorithms for the problem of maximizing a polynomial over the simplex, improving on a result implicit in the 1980 proof of the finite de Finetti theorem by Diaconis and Freedman~\cite{DF80}.  Besides the separable-state approximation problem that we study, hierarchies and nets have been found to have similar performance in finding approximate Nash equilibria~\cite{LMM03, Har-nash} and in estimating the value of free two-prover games~\cite{AIM14, BH-local}. The state-of-the-art run-time for solving Unique Games and Small Set Expansion have also been achieved using both hierarchies and covering-nets. These parallels are summarized in the table:
\begin{table}[h!]
\begin{center}
\begin{tabular}{ccc}
\toprule
Problem & nets & hierarchies/information theory\\
\midrule
$\max_{x\in \Delta_n} p(x)$ & \cite{KlerkLP06} & \cite{DF80,KlerkLP06} \\
approximate Nash & \cite{LMM03,AlonLSV13} & \cite{Har-nash}\\
free games & \cite{AIM14} & \cite[Cor 4]{BH-local} \\
unique games  & \cite{AroraBS10} & \cite{BarakRS11} \\
small-set expansion & \cite{AroraBS10} & \cite[\S 10]{BHKSZ12} \\
separable states & \cite{ShiW12,BKS13}, this work & \cite{BCY10b,
  BH-local, BKS13, LW14, LiS14}\\
\bottomrule
\end{tabular}
\caption{We briefly describe these problems here.  Full descriptions can be found in the references in the table.  In $\max_{x\in \Delta_n} p(x)$, $\Delta_n$ is the $n$-dimensional probability simplex and $p(x)$ is a low-degree polynomial.  ``Approximate Nash'' refers to the problem of finding a pair of strategies in a two-player non-cooperative game for which no player can improve their welfare by more than $\eps$.  ``Free games'' refers to two-prover one-round proof systems where the questions asked are independent; the computational problem is to estimate the largest possible acceptance probability.  ``Unique games'' describes instead proof systems with ``unique'' constraints; i.e.~for each question pair and each answer given by one of the provers, there is exactly one correct answer possible for the other prover.  Small-set expansion asks, given a graph $G$ and parameters $\eps,\delta>0$, whether all subsets with a $\delta$ fraction of the vertices have a $\geq 1-\eps$ fraction of edges leaving the set or whether there exists one with a $\leq \eps$ fraction of edges leaving the set.  Finally ``separable states'' refers to estimating $h_{\Sep(n,n)}$ as we will discuss elsewhere in the paper.  It can also be though of as estimating $\max_{\|x\|_2=1\}}p(x)$ for some low-degree polynomial $p(x)$.}
\end{center}
\end{table}

While this paper focuses on the particular problems where we can improve upon the state-of-the-art algorithms, we hope to be a step towards more generally understanding the connections between these two methods.  In almost every case above, the best covering-net algorithms achieve nearly the same complexity as the best analyses of SDP hierarchies.  There are a few exceptions.  Ref.~\cite{BHKSZ12} shows $O(1)$ rounds of the SoS hierarchy can certify a small value for the Khot-Vishnoi integrality gap instances of the unique games problem, but we do not know how to achieve something similar using nets.  A more general example is in \cite{BKS13}, which shows that the SoS hierarchy can approximate $h_{\Sep}(M)$ in quasipolynomial time when $M$ is entrywise nonnegative.

The closest related paper to this work is \cite{ShiW12} by Shi and Wu (as well as
Appendix A of \cite{BKS13}), which also used enumeration over $\eps$-nets to 
approximate $h_{\Sep}$.  Here we explain their results in our language.

Shi and Wu~\cite{ShiW12} have two algorithms: one when $M$ has low Schmidt rank (i.e. factorizes as $S_1\ra \ell_2^r\ra S_\infty$ for small $r$) and one where $M$ has low rank, which we can interpret as a $\ell_2^r \ra S_\infty^{d_1\times d_2}$ factorization (here $S_\infty^{d_1\times d_2}$ refers to the space of $d_1\times d_2$-dimensional matrices with norm given by the largest singular value).  These correspond to their Theorems 5 and 8 respectively.  In both cases they construct $\eps$-nets for the $\ell_2^r$ unit ball of size $\eps^{-O(r)}$ (here, $\ell_2$ could be replaced with any norm; see Lemma 9.5 of \cite{LT91}).  In both cases, their results can be improved to yield multiplicative approximations, using ideas from \cite{BKS13}.

Appendix A of Barak, Kelner and Steurer~\cite{BKS13} considers fully symmetric 4-index tensors $M\in(\bbR^n)^{\ot 4}$, so that when viewed as $n^2\times n^2$ matrices their rank and Schmidt rank are the same; call them $r$.  Their algorithm is similar to that of \cite{ShiW12}, although they observe additionally (using different terminology) that for any self-adjoint operator $T: \cA^* \ra \cA$ (i.e. satisfying  $\langle T(X), Y\rangle = \langle X, T(Y)\rangle$) the $\cA^* \ra \ell_2 \ra \cA$ norm is equal to the $\cA^*\ra \cA$ norm.  This means that constructing an $\eps$-net for $B(\ell_2)$ actually yields a multiplicative approximation of the $\cA^*\ra \cA$ norm (here the $S_1\ra S_\infty$ norm).

Achieving a multiplicative approximation is stronger than what our algorithms achieve, but it is at the cost of a runtime that can be exponential in the input size even for a constant-factor approximation.  By contrast, our algorithms yield nontrivial approximations in polynomial or quasipolynomial time.

\subsection{Notation}
Define the sets of $d\times d$ real and complex semidefinite matrices
by $\cS^d_+,\cH^d_+$ respectively.  
For complex vector spaces $V,W$, define $\cL(V,W)$ to be the set of linear operators from $V$ to $W$, $\cL(V) := \cL(V,V)$ and $\cH(V),\cH_+(V)$ to be respectively the Hermitian and positive-semidefinite operators on $V$.

For $\alpha\geq 1$ define the
$\ell_\alpha,S_\alpha$ metrics on vectors and matrices respectively by
$\|x\|_{\ell_\alpha} = (\sum_i |x_i|^\alpha)^{1/\alpha}$ and
$\|X\|_{S_\alpha} = (\tr |X|^\alpha)^{1/\alpha}$.  Denote the
corresponding normed spaces by $\ell_\alpha^d, S_\alpha^d$.  Where it
is clear from context we will refer to both norms by
$\|\cdot\|_\alpha$.   We use $\|\cdot \|$ without subscript to denote
the operator norm for matrices (i.e. $\|X\| = \|X\|_{S_\infty}$) and the
Euclidean norm for vectors (i.e. $\|x\| = \|x\|_{\ell_2}$).

We use $\tilde O(f(x))$ to mean $O(f(x)\poly\log(f(x)))$ and say that $f(x)$ is ``quasipolynomial'' in $x$ if $f \leq O(\exp(\poly\log(x)))$.

For a normed space $V$, define $B(V) = \{v \in V : \|v\| \leq 1 \}$.
Two important special cases are the probability simplex $\Delta_n :=
B(\ell_1^n) \cap \bbR_{\geq 0}^n$ and the set of density matrices
(also called ``quantum states'')
$\cD_d := B(S_1^d) \cap \cH^d_+ = \conv\{vv^\dag: v\in B(\ell_2^d)\}$.
Here $v^\dag$ is the conjugate transpose of $v$.  
For $k$ a positive integer, define also
\be \Delta_n(k) := \L\{ \frac{e_{i_1} + \ldots + e_{i_k}}{k} : i_1,\ldots,i_k \
\in [n]\R\} \subset \Delta_n,\ee
where $e_i$ is the vector in $\bbR^n$ with a 1 in position $i$ and
zeros elsewhere.
For a convex set $K$
define the support function $h_K(x) := \sup_{y\in K} \langle
x,y\rangle$.  For matrices $\langle,\rangle$ refers to the
Hilbert-Schmidt inner product $\langle X,Y\rangle := \tr (X^\dag Y)$.

Banach spaces are normed vector spaces with an additional condition (completeness, i.e. convergence of Cauchy sequences) that is relevant only in the infinite dimensional case.  In this work we will consider only finite-dimensional Banach spaces.

\section{Warmup: algorithm for bipartite separability}
\label{sec:warmup}

In this section we describe a simple version of our algorithm.  It
contains all the main ideas which we will later generalize.  Let $M =
\sum_{i=1}^n X_i \ot Y_i$, where $X_i\in \cS^{d_1}_+$, $Y_i\in\cS^{d_2}_+$, $\sum_i X_i \leq
I$, and each $Y_i \leq I$. In quantum information language, $M$ is a 1-LOCC
measurement, meaning it can be implemented with local operations and
one-way classical communication \footnote{Conventionally these have $\sum_i
X_i = I$, but our formulation is essentially equivalent.}.  
 In later sections we will see that
$M$ can also be interpreted in a (mathematically) more natural way as
a bounded map from $S_1$ to $S_\infty$.
The goal of our algorithm is to approximate $h_{\Sep(d_1, d_2)}(M)$,
where we define the set of separable states as
\be \Sep(d_1,d_2) := \conv \{\alpha \ot \beta : \alpha\in
\cD_{d_1},\beta\in\cD_{d_2}\}.\ee

There have been several recent proofs~\cite{BCY10, BH-local, LW14},
each based on quantum information theory, that SDP hierarchies can
estimate $h_{\Sep(d_1,d_2)}(M)$ to error $\eps\|M\|$ in time
$\exp(O(\log^2(d)/\eps^2))$.  Similar techniques also appeared in
\cite{BKS13, LiS14} for different classes of operators $M$.  The role of
the 1-LOCC conditions in these proofs was typically not completely
obvious, and indeed it entered the proofs of
\cite{BCY10, BH-local,LW14} in three different ways. We now give another interpretation of it that is arguably more geometrically natural.

Begin by observing that
\begin{eqnarray}
h_{\Sep(d_1,d_2)}(M) &=& \max_{\substack{\alpha \in \cD_{d_1},\beta\in\cD_{d_2}}}
\tr [(\alpha\ot \beta)M],   \nonumber \\
 &=& \max_{\substack{\alpha \in \cD_{d_1},\beta\in\cD_{d_2}}}
\sum_{i=1}^n \tr [\alpha X_i] \tr [\beta Y_i]  \nonumber \\
&=& \max_{p \in S_X} \|p\|_Y.
\end{eqnarray}

\iffalse
\begin{eqnarray} 
h_{\Sep(d_1,d_2)}(M)
 &=& \max_{\substack{\alpha \in \cD_{d_1},\beta\in\cD_{d_2}}}
\tr [(\alpha\ot \beta)M], \nonumber \\
& =& \max_{\substack{\alpha \in \cD_{d_1},\beta\in\cD_{d_2}}}
\sum_{i=1}^n \tr [\alpha X_i] \tr [\beta Y_i] \nonumber \\
&=& \max_{p \in S_X} \|p\|_Y.
\end{eqnarray}
\fi

In the last step we have defined
\ba \label{defS_xetc}
S_X &:= \{ p \in \Delta_n : \exists \alpha\in \cD_{d_1},\, p_i = \tr
[\alpha X_i]\; \forall i\in [n]\}, 
\qand
%\Delta_n &:= \L\{ p \in \bbR^n : \sum_{i=1}^n p_i =1,\; p_i\geq 0
%\;\forall i\in[n]\R\}, \\
\|a\|_Y := \L\|\sum_{i=1}^n a_i Y_i\R\|.
\ea

The basic algorithm is the following:

\mbox

\begin{mybox}
\begin{algorithm}[Basic algorithm for computing $h_{Sep}\left(M\right)$ for one-way LOCC $M = \sum_i X_i \otimes Y_i$]
\mbox{}\label{alg:basic}
  \begin{description}
  \item[Input:] $\{X_i\}_{i=1}^n \subset \cH_+^{d_1},\{Y_i\}_{i=1}^n \subset \cH_+^{d_2}$. 
  %  Eq. \eq{q-lasserre}.
    %
  \item[Output:] States $\alpha \in \cD_{d_1}$ and $\beta \in \cD_{d_2}$. 
  \end{description}
  \begin{enumerate}
%\item Use \lemref{sparse-basic} to replace $M = \sum_{i=1}^n X_i \otimes Y_i$ with $M' = \sum_{i=1}^{n'} X_i' \ot
%  Y_i'$ satisfying $\|M-M'\|\leq \delta/2$.
\item Enumerate over all $p \in \Delta_n(k)$, with $k =
  9\ln(d_2)/\delta^2$.
\benum \item For each $p$, check (using \lemref{check-close}) whether
there exists $q\in S$ with $\|p-q\|_Y \leq \delta/2$.
\item If so, compute $\|q\|_Y$.
\eenum
\item Let $q$ be such that the $\|q\|_Y$ is the maximum, and let
  $\alpha\in \cD_{d_1}$ be the state for which $q_i  = \tr[X_i
  \alpha]$.  Output this $\alpha$ and $\beta$ satisfying $\tr[\beta
  \sum_i q_i Y_i] = \|\sum_i q_i Y_i\|$.
\eenum
\end{algorithm}
\end{mybox}

\mbox{}

\vspace{0.1 cm}

The main result of this section is:

\begin{thm} \label{basicthm}
Let $M = \sum_{i=1}^n X_i \otimes Y_i$ be such that $\sum_i X_i \leq I$, $X_i \geq 0$, $0 \leq Y_i \leq I$. Algorithm \ref{alg:basic} runs in time $\poly(d_1, d_2, n) \exp\left(O \left(\delta^{-2} \log(n) \log(d_2) \right) \right)$ and outputs $\alpha \in {\cal D}_{d_1}$ and $\beta \in {\cal D}_{d_2}$ such that
\begin{equation}
h_{Sep}(M) \geq \tr[M( \alpha \otimes \beta)] \geq h_{Sep}(M)  - \delta,
\end{equation}
\end{thm}

For $n = \poly(d_1, d_2)$ this is the same running time as found in
\cite{BCY10} (while for $n \ll \poly(d_1, d_2)$ it is an
improvement). Later in this section we will show how we can always
modify the measurement to have $n = \poly(d_1, d_2)$ only incurring in
a small error. But before that, we now show that Theorem
\ref{basicthm} follows easily from two simple lemmas. 

One of the lemmas is a consequence of the the well-known matrix
Hoeffding bound.  

\begin{lem}[Matrix Hoeffding Bound~\cite{Tropp-LD}]\label{lem:Hoeff}
Suppose $Z_1,\ldots,Z_k$ are independent random $d\times d$ Hermitian matrices
satisfying $\E[Z_i]=0$ and $\|Z_i\|\leq \lambda $.  Then
\be \Pr\L[ \L\| \frac{1}{k} \sum_{i=1}^k Z_i \R\| \geq \delta \R] 
\leq d\cdot e^{-\frac{k\delta^2}{8\lambda^2}}.\ee
\end{lem}

This is a special case of Theorem 2.8 from
\cite{Tropp-LD}):

Our first lemma shows that one can restrict the optimization to a net of size $n^{O(\log(d_2)\delta^{-2})}$:

\begin{lem}\label{lem:Y-net}
For any $p\in \Delta_n$ there exists $q\in \Delta_n(k)$ with 
\be \|p-q\|_Y \leq \sqrt{\frac{9\ln(d_2)}{k}} .\label{eq:pq-Y-close}\ee
\end{lem}
\begin{proof}
Sample $i_1,\ldots,i_k$ according to $p$ and set $q = (e_{i_1} +
\ldots +  e_{i_k})/k$.   Define $\bar Y :=
\sum_{i=1}^n p_i Y_i$ and $Z_j = \bar Y- Y_{i_j}$.  Observe that
$\E[Z_j]=0$ and $\|Z_j\|\leq 1$.  Then \lemref{Hoeff} implies that
\be \|p-q\|_Y = \L\| \frac{1}{k} \sum_{j=1}^k  Z_j\R\| \leq \delta.\ee
with positive probability if $k > 8\ln(d)/\delta^2$.  Setting $\delta
= \sqrt{9\ln(d_2)/k}$ we find that there 
exists a choice of $q\in  \Delta_n(k)$ satisfying Eq. \eq{pq-Y-close}.
\end{proof}

The second lemma shows that one can decide efficiently if an element of the net 
is a valid solution.  A similar result is in \cite{ShiW12}.

\iffalse
Shi-Wu has their own version of our checking lemma
(our \lemref{check-close}), which we should cite at the point where
\lemref{check-close} is first stated.  They further show that this type of SDP can
be solved using the multiplicative weight method and therefore can be
solved using either a $\poly\log$-depth, $\poly$-size circuit, or a
$\poly\log$-space, $\poly$-time algorithm.
\fi

%\todo{mention that Shi-Wu use multiplicative weights here}
\begin{lem}\label{lem:check-close}
Given $p\in \Delta_n$ and $\eps>0$, we can decide in time
$\poly(d_1,d_2,n)$ whether the following set is nonempty
\be S_X \cap \{ q : \|p-q\|_Y \leq \eps\}.\ee
\end{lem}
\begin{proof}
Both are convex sets, defined by semidefinite constraints.  So we can
test for feasibility with a SDP of size $\poly(d_1,d_2,n)$. Indeed this is manifest for $S_X$ in Eq. (\ref{defS_xetc}), while $\{ q : \|p-q\|_Y \leq \eps\}$ can be written as
\begin{equation}
\{ q : \|p-q\|_Y \leq \eps\} = \left \{ (q_1, \ldots, q_n) : q_i \geq 0, - \eps I \leq \sum_i p_i Y_i - \sum_i q_i Y_i  \leq \eps I \right \}.
\end{equation}
\end{proof}

We are ready to prove Theorem \ref{basicthm}:

\begin{proof}[Proof of Theorem \ref{basicthm}]
Whatever the output $x$ is, $x \leq h_{\Sep}(M) + \delta/2$.  On
the other hand, let $q = \arg\max_{q\in S} \|q\|_Y$, so that $\|q\|_Y =
h_{\Sep}(M')$.  By \lemref{Y-net}, there exists $p\in \Delta_n(k)$ with
$\|p-q\|_Y \leq \delta/2$.  Thus our algorithm will output a value that is
$\geq h_{\Sep}(M)-\delta$.  We conclude that the
algorithm achieves an additive error of $\delta$ in time
$\poly(d_1, d_2) n^{O(\log(d_2)/\delta^2)}$.
\end{proof}

\subsection{Sparsification}

We now consider the case where $n \gg \poly(d_1, d_2)$. It turns out
that we can modify the algorithm such that its running time is
polynomial in $n$ by first sparsifying the number of local terms of
the measurement.   This results in the following theorem.

\begin{thm} \label{thm:sparse-basic}
Let $M = \sum_{i=1}^n X_i \otimes Y_i$ be such that $\sum_i X_i \leq I$, $X_i \geq 0$, $0 \leq Y_i \leq I$. Algorithm \ref{alg:basic2} runs in time $\poly(n)\exp\left(O \left(\delta^{-2} \log d_1 \log (d_1d_2) \right) \right)$ and outputs $\alpha \in {\cal D}_{d_1}$ and $\beta \in {\cal D}_{d_2}$ such that
\begin{equation}
h_{Sep}(M) \geq \tr(M (\alpha \otimes \beta)) \geq h_{Sep}(M)  - \delta,
\end{equation}
\end{thm}

The key element of the theorem is the following Lemma.

\begin{lem}\label{lem:sparse-basic}
Given a 1-LOCC measurement $M = \sum_{i=1}^n X_i \ot Y_i$ and some
$\eps>0$ there exists
a 1-LOCC measurement $M' = \sum_{j=1}^{n'} X'_j \ot Y'_j$ with
$\|M-M'\|\leq \eps$ and $n' \leq \poly(d_1,d_2)/\eps^2$.  If the
decomposition of $M$ is explicitly given then $M'$ and its
decomposition can be found in time $\poly(d_1, d_2, n)$ using a
randomized algorithm.  
\end{lem}

%In fact, even an implicit specification of $M$ suffices, as we
%describe in Appendix \ref{sparse}. 
%\todo{add}

The modified algorithm is the following:

\mbox

\begin{mybox}
\begin{algorithm}[Algorithm for computing $h_{Sep}\left(M\right)$ for one-way LOCC $M = \sum_i X_i \otimes Y_i$]
\mbox{}\label{alg:basic2}
  \begin{description}
  \item[Input:] $\{X_i\}_{i=1}^n ,\{Y_i\}_{i=1}^n$. 
  %  Eq. \eq{q-lasserre}.
    %
  \item[Output:] States $\alpha \in \cD_{d_1}$ and $\beta \in \cD_{d_2}$. 
  \end{description}
  \begin{enumerate}
\item Use \lemref{sparse-basic} to replace $M = \sum_{i=1}^n X_i \otimes Y_i$ with $M' = \sum_{i=1}^{n'} X_i' \ot
  Y_i'$ satisfying $\|M-M'\|\leq \delta/2$.
\item Run Algorithm 1 on $M'$. 
\end{enumerate}
\end{algorithm}
\end{mybox}

\mbox{}

\vspace{0.1 cm}

The proof of correctness is straightforward.
\begin{proof}[Proof of \thmref{sparse-basic}]
Whatever the output $x$ is, $x \leq h_{\Sep}(M') \leq h_{\Sep}(M) + \delta/2$.  On
the other hand, let $q = \arg\max_{q\in S} \|q\|_Y$, so that $\|q\|_Y =
h_{\Sep}(M')$.  By \lemref{Y-net}, there exists $p\in \Delta_n(k)$ with
$\|p-q\|_Y \leq \delta/2$.  Thus our algorithm will output a value that is
$\geq h_{\Sep}(M')-\delta/2 \geq h_{\Sep}(M)-\delta$.  We conclude that the
algorithm achieves an additive error of $\delta$ in time
$\poly(n) (d_1d_2)^{O(\log(d_2)/\delta^2)}$.
\end{proof}

It remains only to prove \lemref{sparse-basic}.  This requires a careful use of the matrix Hoeffding bound (\lemref{Heoff}).
 The details are in 
Appendix \ref{sparse}.

\subsection{Multipartite}

We now consider the generalization of the problem to the multipartite
case. We consider measurements on a $l$-partite vector space
$\mathbb{C}^{d_1} \otimes \ldots \otimes \mathbb{C}^{d_l}$. Following
Li and Smith~\cite{LiS14}, we define the class of fully one-way LOCC measurements
on $\mathbb{C}^{d_1} \otimes \ldots \otimes \mathbb{C}^{d_l}$
recursively as all measurements $M = \sum_i X_i \otimes M_i$, where
$X_i\in \cH_+^{d_1}$, $\sum_i X_i \leq I$,
$M_i\in\cH_+^{d_2\cdots d_l}$, and each $M_i$ is a fully one-way LOCC measurement in
$\mathbb{C}^{d_2} \otimes \ldots \otimes \mathbb{C}^{d_l}$.

Ref.~\cite{LiS14} recently strengthened the result of \cite{BH-local}
(from parallel one-way LOCC to fully one-way LOCC measurement) and proved that the SoS hierarchy approximates
\begin{equation}
h_{\text{Sep}(d_1, \ldots, d_l)}(M) := \max_{\alpha_1 \in {\cal
    D}_{d_1}, \ldots, \alpha_l \in {\cal D}_{d_l}}
 \tr[(\alpha_1 \otimes \ldots \otimes\alpha_l) M]
\end{equation}
to within additive error $\delta$ in time $\exp(O(\log^2(d)l^3/\delta^2))$, with $d := \max_{i \in [l]} d_i$. Here we show that our previous algorithm for the bipartite case can be extended to the multipartite setting to give  the same run time. 

\begin{thm}
\algref{multi} above runs in time $\exp(O(l^3\ln^2(d)/\delta^2)$ and outputs states $\alpha_i$, $i \in [l]$, satisfying 
\be
h_{Sep(d_1, \ldots, d_l)}\left(M\right) \geq \tr[M (\alpha_1 \otimes \ldots \otimes \alpha_l)]
\geq h_{Sep(d_1, \ldots, d_l)}\left(M\right) - \delta.
\ee
\end{thm}

\mbox

\begin{mybox}
\begin{algorithm}[Algorithm for computing $h_{Sep(d_1, \ldots, d_l)}\left(M\right)$ for fully one-way LOCC $M$]
\mbox{}\label{alg:multi}
  \begin{description}
  \item[Input:] $\{X^{(m)}_{i_1,\ldots,i_m} : m\in [l], i_1\in [n_1],\ldots,i_m\in [n_m]\} \subset \cH_+^{d_m}$ such that
$$ M = \sum_{i_1=1}^{n_1} X^{(1)}_{i_1}
\ot \sum_{i_2=1}^{n_2} X^{(2)}_{i_1,i_2}
\ot \cdots \ot \sum_{i_m=1}^{n_m} X^{(m)}_{i_1,i_2,\ldots,i_m}$$
  %  Eq. \eq{q-lasserre}.
    %
  \item[Output:] States $\alpha_i \in \cD_{d_i}$, $i \in [l]$. 
  \end{description}
  \begin{enumerate}
\item Use \lemref{sparse-basic} to replace $M = \sum_{i_1=1}^{n_1}
  X^{(1)}_{i_1} \otimes M_{i_1}$ with $M' = \sum_{i_1=1}^{n_1'}
  (X^{(1)}_{i_1})' \ot M_{i_1}'$ satisfying $\|M-M'\|\leq \delta/2l$.
  Here $M_{i_1}$ is a shorthand for the collection
  $\{X^{(m)}_{i_1,\ldots,i_m}\}$ for $m\geq 2$ and likewise for
  $M_i'$.  Redefine $M, \{X^{(m)}_{i_1,\ldots,i_m}\}, \{n_i\}$ appropriately.
\item Initialize the variables $\alpha_1,\ldots,\alpha_l$ to $\emptyset$.
\item Enumerate over all $p \in \Delta_n(k)$, with $k = 9l^2\ln(d)/\delta^2$.  For each $p$,
\benum
\item  Check (using \lemref{check-close}) whether
there exists $q\in S_{X^{(1)}}$ with $\|\sum_i (p_i-q_i)M_i\| \leq \delta/2l$. 
\item If no such $q$ exists then do not evaluate this value of $p$ any further.  Otherwise let $\beta_1$ be the density matrix found in the SDP in \lemref{check-close} satisfying $q_i = \tr[\beta_1 X^{(1)}_i]$.
\item For $m'\in \{2,\ldots,m\},i_2\in [n_2],\ldots,i_{m'}\in
  [n_{m'}]$, define 
$\tilde X^{(m'-1)}_{i_2,\ldots,i_{m'}} := \sum_{i_1} q_{i_1} X^{(m')}_{i_1,i_2,\ldots,i_{m'}}$.
\item Recursively call Algorithm~\ref{alg:multi} on input $\{\tilde
  X^{(m')}_{i_1,\ldots,i_{m'}}\}$.  Denote the output by $\beta_2,\ldots,\beta_l$.
\item If $\tr[M(\beta_1\ot \cdots \ot \beta_l)] > \tr[M(\alpha_1\ot
  \cdots \ot \alpha_l)]$ then replace
  $\alpha_1,\ldots,\alpha_l$ with $\beta_1,\ldots,\beta_l$.  
\eenum
\end{enumerate}
\end{algorithm}
\end{mybox}

\mbox{}

\vspace{0.1 cm}

\subsection{The need for an explicit decomposition}\label{sec:explicit}
The input to our algorithm is not only a 1-LOCC measurement $M$ but an explicit
decomposition of the form $M = \sum_i X_i \ot Y_i$ with each $X_i \geq 0$.
Previous algorithms for $h_{\Sep}$ were mostly based on  the
SoS hierarchy (or its restriction to the separability problem also
known as $k$-extendible hierarchy)~\cite{DohertyPS04}.  Running these
requires only knowledge of $M$ and not its decomposition.  The
decomposition appears in the analysis of \cite{BCY10, LW14, BC11,
  BH-local, Yang06}, but not the algorithm.

On the other hand, previous algorithms did not yield an explicit
rounding, i.e. a separable state $\sigma$ with $\tr M \sigma \approx
h_{\Sep}(M)$.  The only exception to this~\cite{BH-local} {\em also}
required an explicit decomposition in order to produce a rounding.

In general any bipartite measurement $M$ can be written in the form
$\sum_i X_i \ot Y_i$, with individual terms that are not necessarily
positive semidefinite.  Finding {\em some} such decomposition is straightforward,
e.g. using the operator Schmidt decomposition or even writing $M =
\sum_{ijkl} M_{ijkl} \ket i \bra j \ot \ket k \bra l$.  Our algorithm
can be readily modified to incorporate non-positive $X_i$ (along the
lines of \secref{inj-norm}), but the run-time will then include a
factor of $\sum_i \|X_i\|_1$ in the exponent.  In general this will be
$O(1)$ only if $M$ is close to 1-LOCC and the decomposition is close to
the correct one.

This raises an interesting open question: given $M$, find a decomposition 
$M = \sum_i X_i \ot Y_i$ that (approximately) minimizes $\sum_i
\|X_i\|_1$.  We are not aware of nontrivial algorithms or hardness
results for this problem.

\section{Generalized algorithm for arbitrary norms}\label{sec:arb-Y-norm}

An important step in the algorithm of the previous section was the identity,
\be 
h_{\Sep(d_1,d_2)}(M) = \max_{p \in S_X} \|p\|_Y,
\ee
valid for any one-way LOCC $M = \sum_{i} X_i \otimes Y_i$. This equation suggests ways of generalizing the algorithm. In this section we consider the setting where the operators $\{ Y_1, \ldots Y_n\}$ belong to some Banach space $\cB$ with norm $\|\cdot\|_{\cB}$. In analogy with Eq. (\ref{defS_xetc}), given $Y = \{ Y_1, \ldots Y_n\}$ we define the $(\cB,Y)$ norm in $\mathbb{R}^n$ as
\begin{equation}
\Vert a \Vert_{\cB, Y} := \left \Vert  \sum_{i=1}^n a_i Y_i  \right \Vert_{\cB}.
\end{equation}

The goal is then to estimate
\be  \label{gen2}
\max_{p \in S_X} \|p\|_{\cB, Y},
\ee
where, as before, $S_X$ is given by Eq. (\ref{defS_xetc}). 

Also this generalization is of interest in quantum information theory. As we discuss more in the next subsection, it includes as a particular case the well-studied problem of computing the maximum output $\alpha$-norms of an entanglement-breaking channel. Consider a general entanglement-breaking quantum channel $\Lambda : {\cal D}_{d_1} \rightarrow {\cal D}_{d_2}$ given by \cite{HSR04}:

%An interesting question for a quantum channel (connected with its ability to transmit  classical information) is how much purity it can preserve. One way of quantifying it  is to compute the maximum output Schatten-$\alpha$ Renyi entropies for all $\alpha > 1$. An important class of quantum channels are the so-called \textit{entanglement-breaking} channels, which are channels that cannot distribute entanglement. 
\begin{equation}
\Lambda(\rho) := \sum_i \tr(X_i \rho) Y_i, 
\end{equation}
with $Y_i \geq 0$, $\tr(Y_i) = 1$, $X_i \geq 0$, and $\sum_i X_i = I$. Then 
\begin{equation}
\max_{\rho \in {\cal D}_{d_1}} \Vert  \Lambda(\rho) \Vert_{\alpha} = \max_{p \in S_X} \|p\|_{S_{\alpha}, Y}.
\end{equation}
%with $S_{\alpha}$ the Schatten-$\alpha$ space with norm $ \Vert  Z \Vert_{\alpha} := \tr(|Z|^{\alpha})^{1/\alpha}$.

In order to find an algorithm for computing Eq. (\ref{gen2}), we need to replace the quantum Hoeffding bound (Lemma \ref{lem:Hoeff}) by more sophisticated concentration bounds. Since in Lemma \ref{lem:check-close} all we needed was a bound in expectation, the right concept will turn out to be the Rademacher type-$\gamma$ constant of the space $\cB$, which we now define:

\begin{dfn} \label{typegamma}
We say a Banach space $\cB$ has Rademacher type-$\gamma$ constant $C$
if for every $Z_1, \ldots, Z_k \in \cB$ and Rademacher random
variables $\varepsilon_1, \ldots, \varepsilon_k$ (i.e. independent and
uniformly distributed on $\pm 1$) , 
\begin{equation} 
\mathop{\mathbb{E}}_{\varepsilon_1, \ldots, \varepsilon_k} \left \Vert \sum_{i=1}^k \eps_i Z_i \right \Vert_{\cB}^\gamma
\leq C^{\gamma} \sum_{i=1}^k  \left \Vert Z_i \right \Vert_\cB^\gamma.
\end{equation}
\end{dfn}

It is known that Schatten-$\alpha$ spaces with norm $\Vert X
\Vert_{\alpha} := \tr(|X|^{\alpha})^{1/\alpha}$ have type-2 constant
$\sqrt{\alpha - 1}$ for $\alpha \geq 2$~\cite{BCL94}, and
type-$\alpha$ constant 1 for every $\alpha \in [1,2]$~\cite[Thm
3.3]{Kato00}.

For a reader unfamiliar with the type-$\gamma$ constant, we suggest verifying
that the type-2 constant of $\ell_2$ is 1.    A more nontrivial
calculation is using the Hoeffding bound or its operator version to
verify that the type-2 constant of $\ell_\infty^n$ or $S_\infty^n$ is
$O(\sqrt{\log n})$.  (This also follows from the fact that the $S_\infty$
and $S_{\log(n)}$ norms are within a constant multiple of each other
on the space of $n$-dimensional matrices.)

For sparsification (the analogue of Lemma \ref{lem:sparse-basic}) we will actually need a slightly stronger condition than a bound on the type-$\gamma$ constant:  

\begin{dfn}
The modulus of uniform smoothness of a Banach space $\cB$ is defined to be the function
\be \rho_{\cB}(\tau) := 
\sup \L\{ \frac{\|x + \tau y\|_{\cB} + \|x - \tau y\|_{\cB}}{2} - 1
\; : \; 
 \|x\|_{\cB} = \|y\|_{\cB} = 1\R\}.
\label{eq:smootheness}\ee
\end{dfn}

By the triangle inequality, $\rho_{\cB}(\tau) \leq \tau$ for all
$\cB$. But when $\lim_{\tau\ra 0}\frac{\rho_{\cB}(\tau)}{\tau}=0$ then
we say that $\cB$ is uniformly smooth. For example, if $\cB =\ell_2$
then $\rho_{\cB}(\tau) = \tau^2/2$, whereas $\rho_{\ell_1}(\tau) =
\tau$. More generally \cite{BCL94} (building on \cite{Tomczak1974})
proved that $\rho_{S_\alpha}(\tau) \leq \frac{\alpha-1}{2}\tau^2$ for
$\alpha>1$. We say that $\cB$ has modulus of smoothness of power type
$\gamma$ if $\rho_{\cB}(\tau) \leq C\tau^\gamma$ for some constant
$C$. This implies (using an easy induction on $k$) that the type-$\gamma$ constant is
$\leq C$, and indeed this was how the type-$\gamma$ constant was
bounded in \cite{Tomczak1974,BCL94}.

The algorithm for approximating the optimization problem given by
Eq. (\ref{gen2}) is the following:

\mbox

\begin{mybox}
\begin{algorithm}[Algorithm for computing $\max_{p \in S_X} \|p\|_{\cB, Y}$ for $\cB$ of type-$\gamma$ constant $C$ and modulus of uniform smoothness $\rho_{\cB}(\tau) \leq s \tau^2$, with $X := \{ X_i \}$ and $Y := \{ Y_i \}$]
\mbox{}\label{alg:1stgen}
  \begin{description}
  \item[Input:] $\{X_i\}_{i=1}^n ,\{Y_i\}_{i=1}^n$
  %  Eq. \eq{q-lasserre}.
    %
  \item[Output:] $p \in {\cal S}$
  \end{description}
  \begin{enumerate}
\item Use Lemma \ref{lem:sparse-general}  to replace $X = \{ X_i \}$ and $Y = \{ Y_i \}$ with $X' := \{ X'_i \}$ and $Y' := \{ Y'_i \}$.
\item Enumerate over all $p \in \Delta_n(k)$, with $k = \left( \frac{2C^{\gamma}}{\delta^{\gamma}} \max_{i} \Vert Y_i \Vert_{\cB}^{\gamma}  \right)^{1/(\gamma-1)}$.
\benum \item For each $p$, check (using Lemma \ref{checkingefficinetly2}) whether
there exists $q\in S$ with $\|p-q\|_{\cB, Y} \leq \delta$.
\item If so, compute $\|p\|_Y$.
\eenum
\item Output $p$ such that $\|p\|_{\cB, Y}$ is the maximum.
\eenum
\end{algorithm}
\end{mybox}

\mbox{}

We have:

\begin{thm} \label{thm2}
Let $\cB$ be a Banach space with norm $\|\cdot\|_{\cB}$.  Suppose the
type-$\gamma$ constant of $\cB$ is $C$ and that there is $s > 0$ such that the modulus of uniform smoothness satisfies $\rho_{\cB}(\tau) \leq s \tau^2$. Suppose one can compute $\|\cdot\|_{\cB}$ in time $T$. Consider $\{ X_i \}_{i=1}^n$ with $X_i$ $d \times d$ matrices satisfying  $X_i \geq 0$, $\sum X_i \leq I$, and $\{ Y_i \}_{i=1}^{n}$ with $Y_i \in \cB$. Algorithm \ref{alg:1stgen} runs in time 
\be
\poly(T, d, s) \exp \left( O \left(\left( C\delta^{-1} \max_{i} \Vert Y_i \Vert_{\cB}  \right)^{\frac{\gamma}{\gamma-1}} \log(d)  \right) \right) 
\ee
and outputs $p$ such that
\begin{equation}
\max_{p \in S_X} \|p\|_{\cB, Y} \geq \|p\|_{\cB, Y} \geq \max_{p \in S_X} \|p\|_{\cB, Y}  - \delta,
\end{equation}
\end{thm}

As an example, suppose $\cB$ is ${\cal S}_{\infty}^{d_2}$. Then the type-2 constant is $O(\sqrt{\log(d_2)})$, $\max_i\Vert Y_i \Vert \leq 1$, and Theorem \ref{basicthm} shows one can compute $\max_{p \in S_X} \|p\|_{Y}$ in time $\exp(O(\delta^{-2} \log(d_1)\log(d_1d_2)))$. 

In the next subsection we discuss a few particular cases of the
theorem worth emphasizing. Then we prove the theorem.

\subsection{Consequences of Theorem \ref{thm2}}\label{sec:consequences}

\subsubsection{Restricted one-way LOCC measurements}  \label{lowrank}

The next lemma shows that for subclasses of one-way LOCC measurements one has a PTAS for computing $h_{Sep}$. The class include in particular one-way LOCC measurements in which Bob's measurements are low rank.

\begin{cor} \label{cor:lowrankBob}
Let $M = \sum_{i} X_i \otimes Y_i$ be such that $X_i \geq 0$, $\sum_i X_i \leq I$ and $\Vert Y_i \Vert_{2} \leq r$. Then one can compute $\alpha \in {\cal D}_{d_1}$ and $\beta \in {\cal D}_{d_2}$ such that
\begin{equation}
 h_{Sep}(M) \geq \tr(M( \alpha \otimes \beta)) \geq h_{Sep}(M) - \delta
\end{equation}
in time $d_1^{O(\delta^{-2}r)}$.
\end{cor}
\begin{proof}
We use Theorem \ref{thm2} and Algorithm \ref{alg:1stgen} to estimate the optimal $p$ and then find $\alpha$ and $\beta$ by semidefinite programming.
\end{proof}

If instead we use the multipartite version of the algorithm (see Algorithm \ref{alg:multi}), we find that for $M = \sum_{i} X_i \otimes Y_{i_1} \otimes \ldots \otimes Y_{i_l}$, with  $X_i \geq 0$, $\sum_i X_i \leq I$ and $\Vert Y_i \Vert_{2} \leq r$, we can compute $\alpha \in {\cal D}_{d}$ and $\beta_1, \ldots, \beta_l \in {\cal D}_{d}$ such that
\begin{equation}
h_{Sep}(M) \geq \tr(M \alpha \otimes \beta_1 \otimes \ldots \otimes \beta_l) \geq h_{Sep}(M) - \delta
\end{equation}
in time $d^{O(\delta^{-2}l^3 r)}$.

\subsubsection{Maximum output norm of entanglement-breaking channels}

The next corollary shows that for all $\alpha > 1$, there is a PTAS for computing the maximum output Schatten-$\alpha$ norm of an entanglement-breaking channel.

\begin{cor} \label{cor:maxoutputentropy}
Let $\Lambda : {\cal D}_{d_1} \rightarrow  {\cal D}_{d_2}$ be an entanglement-breaking channel with decomposition $\Lambda(\rho) := \sum_i \tr(X_i\rho) Y_i$ (where $X_i \geq 0$, $\sum_i X_i = I$, $Y_i \in {\cal D}_{d_2}$). 
\begin{enumerate}

\item  For every $\alpha \geq 2$ one can compute in time $\poly(d_2) d_1^{O\left( \delta^{-2} \alpha \right)}$ a number $r$ such that
\begin{equation}
\max_{\rho \in {\cal D}_{d_1}} \left \Vert \Lambda(\rho) \right \Vert_{\alpha} \geq r \geq \max_{\rho \in {\cal D}_{d_1}} \left \Vert \Lambda(\rho) \right \Vert_{\alpha} - \delta,
\end{equation}

\item For every $1 < \alpha \leq 2$ one can compute in time $\poly(d_2)  d_1^{O\left( \left( \alpha \delta^{-\alpha}  \right)^{\frac{1}{\alpha - 1}}   \right)}$ a number $r$ such that
\begin{equation}
\max_{\rho \in {\cal D}_{d_1}} \left \Vert \Lambda(\rho) \right \Vert_{\alpha} \geq r \geq \max_{\rho \in {\cal D}_{d_1}} \left \Vert \Lambda(\rho) \right \Vert_{\alpha} - \delta,
\end{equation}
\end{enumerate}
\end{cor}

\begin{proof}
Part 1 follows from Theorem \ref{thm2} and the fact that $S_{\alpha}$, with $\alpha \geq 2$, has type-2 constant $\sqrt{\alpha - 1}$ \cite{BCL94} and $\rho_{S_\alpha}(\tau) \leq \frac{\alpha-1}{2}\tau^2$ for
$\alpha>1$. Part 2, in turn, follows from Theorem \ref{thm2} and the fact that for $S_{\alpha}$, with $\alpha \geq 2$, has type-$\alpha$ constant one \cite[Thm
3.3]{Kato00}.  
\end{proof}

We note that computing maximum output $\alpha$-norms for general
quantum channels is harder. In particular it was shown in
\cite{HM10,HM13} that there is no algorithm that run in time $\exp \left( O
  \left( \log^{2 - \varepsilon}{d} \right) \right )$ for any
$\varepsilon > 0$ and can decide if $\max_{\rho} \Vert \Lambda(\rho)
\Vert$ is one or smaller than $\delta$ (for any fixed $\delta > 0$)
for a general quantum channel $\Lambda : {\cal D}_{d} \rightarrow {\cal D}_{d}$, unless the exponential time
hypothesis (ETH) is wrong (meaning there is a subexponential time
algorithm for 3-SAT).   

The result of \cite{HM10} is one example of many that found
$d^{\tilde\Theta(\log d)}$ upper or lower bounds for related optimization
problems \cite{LMM03, BKW14, HM10}. In
a few cases \cite{Alon2013, ShiW12, de2006ptas} poly-time approximate schemes (PTASs) are known.  Our results here
fall into this second class. We hope that the geometric perspective
from our paper can lead to a better understanding of what
distinguishes these cases.

What is known about hardness results for entanglement-breaking channels? Using the results of \cite{BHKSZ12} one can show that to determine if $\max_{\rho} \Vert \Lambda(\rho) \Vert$ is $\geq C/d$ or $\leq c/d$ (for any two constants $C > c > 0$) cannot be done in time $\exp \left( O \left( \log^{2 - \varepsilon}{d} \right) \right )$ assuming ETH. So one cannot hope to find a polynomial-time algorithm for a \textit{multiplicative} approximation of the maximum output norm. 

Note that the complexity of the algorithm blows up when $\alpha
\rightarrow 1$. This is not only an artifact of the proof. Computing
the quantity for $\alpha$ close to one allow us to estimate the von
Neumann minimum output entropy of the channel. However to estimate it
we need a number of samples of order $O(d)$ and so the net-based approach we
explore in this paper does not lead to efficient algorithms.

\subsubsection{Hypercontractive norms}

Our third corollary concerns the problem of computing hypercontractive norms, in particular computing the $2 \rightarrow s$ norm of a $d \times d$ matrix $A$, for $s > 2$, defined as
\begin{equation}
\Vert A \Vert_{2 \rightarrow s} := \max_{\Vert x \Vert_2 = 1} \Vert A x \Vert_{s}.
\end{equation}
This norms are important in several applications, e.g. bounding the mixing time of Markov chains and determining if a graph is a small-set expander \cite{BHKSZ12}. In \cite{BHKSZ12} it was also shown that to compute any constant-factor multiplicative approximation to the $2 \rightarrow 4$ norm of a $n \times n$ matrix is as hard as solving 3-SAT with $O(\log^2(n))$ variables. In Appendix \ref{hardness} we extend the approach of \cite{BHKSZ12} to show hardness results for multiplicatively approximating all $2 \rightarrow q$ norms, for even $q \geq 4$. 

In \cite{BHKSZ12} it was shown that the result of \cite{BCY10} implies that for any $d \times d$ matrix $A$ the Sum-of-Squares hierarchy computes in time $d^{O(\log(d)\delta^{-2})}$ an additive approximation $x$ s.t.
\begin{equation} \label{Baraketalapprox}
\Vert A \Vert_{2 \rightarrow 4}^4 \leq x \leq \Vert A \Vert_{2 \rightarrow 4}^4  + \delta \Vert A \Vert_{2 \rightarrow 2}^2 \Vert A \Vert_{2 \rightarrow \infty}^2,
\end{equation}
where $\Vert A \Vert_{2 \rightarrow 2}$ is the largest singular value of $A$ and $\Vert A \Vert_{2 \rightarrow \infty}$ the largest 2-norm of any row of $A$.

Using Theorem \ref{thm2} we can improve this algorithm in two ways: First we can compute an approximation to $\Vert . \Vert_{2\rightarrow s}$ for any $s>2$. Second the running time for fixed error is polynomial, instead of quasipolynomial. 

\begin{cor} \label{cor:2to4}
For any $s \geq 2$ one can compute in time $d^{O(s \delta^{-2} )}$ a number $x$ such that
\begin{equation}
\Vert A \Vert_{2 \rightarrow s}^2 \geq x \geq \Vert A \Vert_{2 \rightarrow s}^2 - \delta \Vert A \Vert_{2 \rightarrow 2}^2.
\end{equation}
\end{cor}

\begin{proof}
Let $X_{i} := A^\dag \ket{i} \bra{i} A / ||A||_{2 \rightarrow 2}^2$. Note $X_i \geq 0$ and $\sum_i X_i \leq I$. We can write
\begin{equation}
\Vert A \Vert_{2 \rightarrow s}^{s} = \max_{\ket{\psi} \in \ell_2} \sum_i \bra{\psi} A^{\cal y} \ket{i}\bra{i} A \ket{\psi}^{s/2} = ||A||_{2 \rightarrow 2}^{s} \max_{p \in S_X} \Vert p \Vert_{s/2}^{s/2}. 
\end{equation}

Since $\ell_s$ has type-2 constant $\sqrt{s-1}$, by Theorem \ref{thm2} we can estimate
\be
\max_{p \in S_X} \Vert p \Vert_{s/2} = \frac{\Vert A \Vert_{2 \rightarrow s}^2}{\Vert A \Vert_{2 \rightarrow 2}^2}
\ee
in time $\exp\left( O \left( s \delta^{-2}  \log(d) \right)  \right)$ with additive error $\delta$.
\end{proof}

Although the corollary above gives an approximation for every $s > 2$ that can be computed in polynomial time for every fixed error, it gives a worse approximation to the $2 \rightarrow 4$ than \cite{BHKSZ12} (given by Eq. (\ref{Baraketalapprox})). We now show a second corollary that strictly improves the result of \cite{BHKSZ12} for $2 \rightarrow 4$ and generalizes it to $2 \rightarrow s$ norms for every even $\geq 4$. 

\begin{cor}
For any even $s \geq 4$ one can compute in time $d^{O(s \delta^{-2})}$ a number $x$ such that
\begin{equation}
\Vert A \Vert_{2 \rightarrow s}^s \geq x \geq \Vert A \Vert_{2 \rightarrow s}^s - \delta  \|A\|_{2\ra 2}^2 \|A\|_{2\ra \infty}^{s-2}.
\end{equation}
\end{cor}

\begin{proof}
Define
\be 
X_{i} := \frac{A^\dag \ket{i} \bra{i} A}{||A||_{2\ra 2}^2}
\qand
Y_i := \L(\frac{A^\dag \ket{i} \bra{i} A}{||A||_{2\ra
    \infty}^2}\R)^{\otimes \frac{s}{2}-1}.\ee
Observe that $X_i, Y_i\geq 0$, $\sum_i X_i \leq I$ and $Y_i \leq I$.
Additionally
\be \|A\|_{2\ra s}^s = \|A\|_{2\ra 2}^2 \|A\|_{2\ra \infty}^{s-2}
h_{\Sep^{s/2}(n)}(\sum_i X_i \ot Y_i)
\ee
This last term can be approximated to additive error $\eps$ in time 
$$\exp(O(s^3/\eps^2))$$ using the multipartite results of Section \ref{lowrank}.
\end{proof}

\subsection{Proof of Theorem \ref{thm2}}\label{sec:proof-type-2}

The proof of Theorem \ref{thm2} will follow from three lemmas, the first showing that it is enough to search over a net of small size, the second showing that one can decide membership of $\{ q : \|p-q\|_{\cB, Y} \leq \delta \}$ efficiently (assuming that $\Vert .\Vert_{\cB}$ can be computed efficiently), and the third giving a sparsification for the number of $\{  X_i\}_{i}^n$ and $\{  Y_i \}_{i=1}^n$. 

We first show how the type-$\gamma$ constant gives a concentration bound. This uses a standard argument.

\begin{lem}[Symmetrization Lemma]  \label{symmetrization}
Suppose we are given $p\in\Delta_n$, $ Z_i, \ldots, Z_n$ elements of a Banach space $\cB$ with norm $\|\cdot\|_{\cB}$, and $\varepsilon_1, \ldots, \varepsilon_k$  Rademacher distributed random variables. Then for every $\gamma \geq 1$
\begin{equation}
\mathop{\mathbb{E}}_{i_1, \ldots, i_k \sim p^{\otimes k}} \left \Vert  \frac{1}{k} \sum_{j=1}^k Z_{i_j} - \mathop{\mathbb{E}}_{i \sim p} Z_{i}\right \Vert_{\cB}^{\gamma} \leq 2 \mathop{\mathbb{E}}_{i_1, \ldots, i_k \sim p^{\otimes k}} \mathop{\mathbb{E}}_{\varepsilon_1, \dots, \varepsilon_k} \left \Vert  \frac{1}{k} \sum_{j=1}^k \varepsilon_{j} Z_{i_j}     \right \Vert_{\cB}^{\gamma}. 
\end{equation}
\end{lem}

\begin{proof}
\ba
\E_{i_1,\ldots,i_k \sim p^{\otimes k}}\L\|\frac 1 k\sum_{j=1}^k (Z_{i_j} - \E_{i \sim p} Z_i    )\R\|_{\cB}^{\gamma}
& =\E_{i_1,\ldots,i_k \sim p^{\otimes k} }\L\|\frac 1 k\sum_{j=1}^k (Z_{i_j} - \E_{i_j' \sim p} [Z_{i_j'}])\R\|_{\cB}^\gamma
\\
& \leq \E_{i_1,\ldots,i_k \sim p^{\otimes k}}\E_{i_1',\ldots,i_k' \sim p^{\otimes k}}
\L\|\frac 1 k\sum_{j=1}^k (Z_{i_j} - Z_{i_j'})\R\|_{\cB}^\gamma
\\ &= 
\E_{i_1,\ldots,i_k \sim p^{\otimes k}}\E_{i_1',\ldots,i_k' \sim p^{\otimes k}}\E_{\eps_1,\ldots,\eps_k}
\L\|\frac 1 k\sum_{j=1}^k \eps_j(Z_{i_j} - Z_{i_j'})\R\|_{\cB}^\gamma \\
& \leq 2 \E_{i_1,\ldots,i_k \sim p^{\otimes k}}\E_{\eps_1,\ldots,\eps_k}
\L\|\frac 1 k\sum_{j=1}^k \eps_jZ_{i_j}\R\|_{\cB}^\gamma. \\
\ea
\end{proof}

Then we have the following generalization of Lemma \ref{lem:check-close}:

\begin{lem}\label{lem:BY-net}
Let the Banach space $\cB$ have type-$\gamma$ constant $C$. Then for any $p\in \Delta_n$ there exists $q\in N_k$ with 
\be 
\|p-q\|_{\cB, Y} \leq \left( \frac{2C^{\gamma}}{k^{\gamma - 1}} \mathop{\mathbb{E}}_{i \sim p} \left \Vert Y_i \right \Vert_{\cB}^{\gamma} \right)^{1/\gamma}.
\label{eq:pq-BY-close}
\ee
\end{lem}
\begin{proof}
Sample $i_1,\ldots,i_k$ according to $p$ and set $q = (e_{i_1} +
\ldots +  e_{i_k})/k$.  Then Definition \ref{typegamma} and Lemma \ref{lem:BY-net} give 
\begin{eqnarray}
\left( \mathop{\mathbb{E}}_{i_1, \ldots, i_k \sim p^{\otimes k}}  \|p-q\|_{\cB, Y} \right)^{\gamma}   &\leq& \mathop{\mathbb{E}}_{i_1, \ldots, i_k \sim p^{\otimes k}}  \|p-q\|_{\cB, Y}^{\gamma}  \nonumber \\
&=& \mathop{\mathbb{E}}_{i_1, \ldots, i_k \sim p^{\otimes k}} \left \Vert  \frac{1}{k} \sum_{j=1}^k Y_{i_j} - \mathop{\mathbb{E}}_{i \sim p} Y_{i}\right \Vert_{\cB}^{\gamma}  \nonumber \\
&\leq&    2 \mathop{\mathbb{E}}_{i_1, \ldots, i_k \sim p^{\otimes k}} \mathop{\mathbb{E}}_{\varepsilon_1, \dots, \varepsilon_k} \left \Vert  \frac{1}{k} \sum_{j=1}^k \varepsilon_{j} Y_{i_j}     \right \Vert_{\cB}^{\gamma}.             \nonumber \\
&\leq&  \frac{2C^{\gamma}}{k^{\gamma}}  \mathop{\mathbb{E}}_{i_1, \ldots, i_k \sim p^{\otimes k}}   \sum_{i=1}^k   \left \Vert Y_{i_j} \right \Vert_\cB^\gamma \nonumber \\
&=& \frac{2C^{\gamma}}{k^{\gamma - 1}}  \mathop{\mathbb{E}}_{i \sim p}  \left \Vert Y_i \right \Vert_\cB^\gamma.
\end{eqnarray}
The first inequality follows from the convexity of $x \mapsto x^{\gamma}$, the second inequality from Lemma \ref{symmetrization}, and the third from the fact that $\cB$ has type-$\gamma$ constant $C$.
\end{proof}

The next lemma is an analogue of Lemma \ref{lem:check-close}:

\begin{lem}  \label{checkingefficinetly2}
Let the Banach space $\cB$ be such that $\Vert . \Vert_{\cB}$ can be computed in time $T$. Given $p\in \Delta_n$ and $\eps>0$, we can decide in time $\poly(T, d, n)$ whether the following set is nonempty
\be S_X \cap \{ q : \|p-q\|_{\cB, Y} \leq \eps\} \label{eq:close-set}\ee
\end{lem}
\begin{proof}
Since we have an efficient algorithm for $\|\cdot\|_{\cB}$ we can
efficiently test membership in the set $\{ q : \|p-q\|_{\cB, Y} \leq
\eps\}$.  Thus we can determine if Eq. \eq{close-set} is nonempty using
the ellipsoid algorithm~\cite{GLS}.
\end{proof}

We now state an analogous sparsification result of Lemma
\ref{lem:sparse-basic} for the more general case we consider in this
section.  The proof is in Appendix \ref{sparse}.

\begin{lem}\label{lem:sparse-general} 
Suppose $\Lambda$ is a map from $d\times d$ Hermitian matrices to a
Banach space $\cB$ and is given by $\Lambda(\rho) =
\sum_{i=1}^n \langle X_i, \rho\rangle Y_i $ where each $X_i \geq
0$, $\sum_{i=1}^n X_i \leq I$ and each $\|Y_i\|_{\cB}\leq 1$.  
Suppose
  that $\cB$ has modulus of smoothness $\rho_\cB(\tau) \leq s\tau^2$.
Then
there exists $\Lambda'$ such that $\Lambda'(\rho) =
\sum_{i=1}^{k}  \langle X_i', \rho\rangle Y_i'$ where each $X_i' \geq
0$, $\sum_{i=1}^k X_i' \leq I$ and each $\|Y_i'\|_{\cB}\leq 1$.
Additionally $k \leq c d^2 (d+s) / \delta^2$ for some constant $c>0$, 
\be %\|\Lambda - \Lambda'\|_{S_1 \ra \cB}:= 
\max_{\rho \in \cD_d} \| (\Lambda' - \Lambda)(\rho)\|_{\cB}
\leq \delta
\label{eq:SB-approx}\ee
and $\Lambda'$ can be found efficiently.
\end{lem}

With the lemmas in hand the proof of Theorem \ref{thm2} follows along the same lines as \thmref{sparse-basic}.

\section{Algorithm for injective tensor norm}  \label{sec:inj-norm}
In this section we present one further generalization, this time on
the input space.  While this final generalization does not have
natural applications in quantum information (to our knowledge), it
does give perspective on why it is natural to consider 1-LOCC
measurements and entanglement-breaking channels.

First, we introduce some more definitions.
Suppose that
$\|\cdot\|_{\cA}$ and $\|\cdot\|_{\cB}$ are two norms.  For $\Lambda$ an
operator from $\cA\ra \cB$ define the operator norm
\be \|\Lambda\|_{\cA \ra \cB} := \sup_{a\in B(\cA)} \|\Lambda (a)\|_{\cB}.\ee
Define the injective tensor norm
$\cA \ot_{\inj} \cB$ by
\be \|x\|_{\cA \ot_{\inj} \cB} = \sup_{\substack{a\in B(\cA^*)\\ b\in
    B(\cB^*)}} \langle a\ot b, x\rangle.\ee
Here $\cA^*$ is the space of functions from $\cA$ to $\bbR$, and
$\|\tilde a\|_{\cA^*} := \sup_{a\in B(\cA)} \tilde a(a)$.  
For example, if $\cA=\cB = \ell_2$ then $\cA \ot_{\inj} \cB$ is the
usual operator norm for matrices, i.e. largest singular value.  More
generally $\cA^* \ot_{\inj} \cB$ is isomorphic to the operator norm on
maps from $\cA \ra \cB$.   Finally if $\cA, \cB,\cC$ are Banach spaces
then define the factorization norm $\cA\ra \cB\ra \cC$ for $x\in
\cL(\cA,\cC)$ by
\be \|\Lambda\|_{\cA\ra \cC\ra \cB} = 
\inf_{\substack{\Lambda_1\in \cL(\cC, \cB)\\ \Lambda_2\in \cL(\cA,\cC) \\
\Lambda = \Lambda_1 \Lambda_2}} \|\Lambda_1\|_{\cC\ra \cB}
\|\Lambda_2\|_{\cA \ra \cC}.\ee 
We can now (informally) state our generalized estimation theorem.  Given an
operator $\Lambda \in B(\cA \ra \ell_1 \ra \cB)$ we can estimate $\|\Lambda\|_{\cA
  \ra \cB}$ efficiently.

For example, consider $h_{\Sep}$, which we considered in
\secref{warmup}.  In our new notation 
\be h_{\Sep}(M) = \|M\|_{S_\infty \ot_{\inj} S_\infty}
= \|\hat M\|_{S_1\ra S_\infty},\ee
where $\hat M$ is the map defined by $\hat M(X) = \tr_A[M(X\ot I)]$.
The requirement that $M$ is 1-LOCC is roughly equivalent to the
requirement that
\be \|\hat M\|_{S_1 \ra \ell_1 \ra S_\infty} \leq 1.\ee

\begin{thm} \label{generalthm}
  Suppose $\cA,\cB$ are $d$-dimensional Banach spaces.  Suppose
  $\|\Lambda\|_{\cA \ra \ell_1^n \ra \cB} \leq 1$ and that a good
  factorization is known; i.e. $x_1^*,\ldots,x_n^*\in \cA^*$ and
  $y_1,\ldots,y_n\in \cB$ are given such that $\Lambda = \sum_{i=1}^n
  y_ix_i^*$, $\sup_{a\in \cA} \sum_{i=1}^n |x_i^*(a)|\leq 1$ and 
  $\max_i \|y_i\|_{\cB} \leq 1$.  Suppose further that algorithms exist for
  computing the $\cA$ and $\cB$ norms running in times $T_{\cA},
  T_{\cB}$ respectively.  Let $\lambda$ denote the type-$\gamma$
  constant of $\cB$.  Then we can estimate $\|\Lambda\|_{\cA \ra \cB}$ to
  accuracy $\eps$ in time
\be T_{\cA} T_{\cB} \poly(d) n^{c(\lambda/\delta)^{\frac{\gamma}{\gamma-1}}}.\ee
\end{thm}

The algorithm follows similar lines to the  earlier algorithms.  It lacks
only the sparsification step since we do not know how to extend
\lemref{sparse-general} to this case.

\mbox

\begin{mybox}
\begin{algorithm}[Algorithm for computing $\|\Lambda\|_{\cA \ra \cB}$]
\mbox{}\label{alg:general}
  \begin{description}
  \item[Input:] $\{x_i\}_{i=1}^n ,\{y_i\}_{i=1}^n$
  \item[Output:] $p \in {\cal S}$
  \end{description}
  \begin{enumerate}
\item Enumerate over all $p \in N_k$, with $k = (2\lambda/\delta)^{\frac{\gamma}{\gamma-1}}$. 
\benum \item For each $p$, check whether
there exists $q\in S$ with $\|p-q\|_{\cB, Y} \leq \delta$.
\item If so, compute $\|p\|_Y$.
\eenum
\item Output $p$ such that $\|p\|_{\cB, Y}$ is the maximum.
\eenum
\end{algorithm}
\end{mybox}

\mbox{}

The proof of Theorem \ref{generalthm} is almost the same as that of  Theorem~\ref{thm2}.  The only
new ingredient is checking whether $p \in S_X$.  This is equivalent to
asking whether $\exists a \in B(\cA)$ such that $p_i = x_i^*(a)$.
This is a convex program which can be decided in time $\poly(d)T_{\cA}$ using the
ellipsoid algorithm along with our assumption that $\|\cdot\|_{\cA}$
can be computed in time $T_{\cA}$.

%\section{rest}

%\bit
%\item $h_{\Sep}$ for 1-LOCC measurements with $\|B_i\|_2\leq 1$,
%  treating $M$ as a map from  $S_1\ra S_2$.
%\item getting $2\ra \infty$ norm into $2\ra 4$ norm bound.
%\item $h_{\Sep}(M)$ for $\|M\|_2 \leq 1$.  Explain one of the Shi-Wu
 % results in our framework.
%\item Multiparty.  Use fact that $h_{\Sep}$ inherits concentration from operator norm.  conjecture that $h_{\Sep}$ has concentration in terms of $h_{\Sep}$.  %Think about implications of this conjecture.  Think about implications of this algorithm for channel problems.
%\item compare with n-c Khintchine. 
%\item case of small $h_{\Sep}$.
%\item table of previous proofs
%\eit

%things to ask Aubrun or Szarek or someone about:  When do we have op Chernoff bounds?  We think it's in terms of type-2 constant.  What would be necessary to show concentration for $h_{\Sep}$?  Relation to mean width?

\section*{Acknowledgments}
We thank Jop Briet, Pablo Parrilo and Ben Recht for interesting
discussions. FGSLB is supported by EPSRC. AWH was funded by NSF grants
CCF-1111382 and CCF-1452616, ARO contract W911NF-12-1-0486 and a Leverhulme Trust Visiting Professorship VP2-2013-041. Part of this work was done while A.W. was visiting UCL. 

\appendix
\section{Sparsification} \label{sparse} 

In this appendix we prove two Lemmas about sparsification: one (\lemref{sparse-basic}) for the problem of $h_{\Sep}$ and the second (\lemref{sparse-general}) for the estimate $S_1\ra \cB$ norms.  While the former is a special case of the latter, it is also far more self-contained (requiring only the operator Hoeffding bound), so we recommend reading it first.

\begin{proof}[Proof of \lemref{sparse-basic}]
Assume initially that each $\|Y_i\|=1$.  This is possible because we
can always drop terms with $Y_i=0$ and then rewrite $M$ as
\be \sum_{i=1}^n \|Y_i\|X_i \ot \frac{Y_i}{\|Y_i\|}.\ee
Redefining $X_i$ appropriately we see that $\sum_i
X_i \leq I$ still holds.

Now write $M = \sum_{i=1}^n p_i W_i$ with $p_i = \frac{\tr(X_i \ot Y_i)}{\tr(M)}$
and $W_i = \frac{X_i \ot Y_i}{p_i}$.    
%Since $M \leq I$ we have $\sum_i p_i\leq 1$.
Sample $i_1,\ldots,i_{n'}$ according to $p$ and define
\be A = \sum_{j=1}^{n'} \frac{W_{i_j}}{n'}
\qand
B = \sum_{j=1}^{n'} \frac{X_i/p_i}{n'}.\ee
We would like to guarantee that  
\begsub{guarantees}
\|A - M\| & \leq \delta  \label{eq:guar-A}\\
\|B - I_{d_1}\| & \leq \delta \label{eq:guar-B}
\endsub
for some $\delta$ to be chosen later.  We can use \lemref{Hoeff}
here.  To do so, note that 
\begsub{norm-bounds}
\|W_i\| &\leq \tr W_i \leq \tr M \leq d_1 d_2 \\ 
\|X_i/p_i\| & \leq \frac{\tr M}{\tr Y_i} 
\leq \tr M \leq d_1 d_2, \text{ using the assumption that }\|Y_i\|=1
\endsub
Now we find that the probability that Eq. \eq{guarantees} fails to hold is
\be \leq d_1d_2 \exp\L(-\frac{n' \delta^2 }{8 d_1^2d_2^2}\R)
+ d_1 \exp\L(-\frac{n' \delta^2 }{8 d_1^2d_2^2}\R).\ee
Taking $n' = 8d_1^2d_2^2\log(2d_1d_2)/\delta^2$ we have that
\eq{guarantees} holds with positive probability.  Fix the
corresponding $i_1,\ldots,i_{n'}$.  Choose $M' = A/(1+\delta)$.
Together with Eq. \eq{guarantees} this means that $M'$ is a valid 1-LOCC
measurement.  By Eq. \eq{guar-A} we can achieve our result by choosing
$\delta = \eps/3$.  Indeed
\begin{eqnarray}
\| M' - M\| &=& \L\| \frac{A}{1+\delta} -  M \R\| \nonumber \\
&\leq& \L\| \frac{A}{1+\delta} -  A \R\| + \L\| A -  M \R\| \nonumber \\
%\leq \L(1-\frac{1}{1+\delta}\R)\|A\| + \delta
&\leq& \L(1-\frac{1}{1+\delta}\R)(1+\delta) + \delta \nonumber \\
&\leq& 2\delta + \delta^2 \leq \eps.
\end{eqnarray}

\end{proof}

We now turn to the proof of \lemref{sparse-general}, covering the case of general Banach spaces with bounded modulus of smoothness.

We will need the following Azuma-type inequality from
Naor~\cite{Naor}, who attributes it to Pisier.  We will state a weaker
Hoeffding-type formulation that suffices for our purposes.
\begin{lem}[Theorem 1.5 of \cite{Naor}]\label{lem:azuma}
Suppose $X_1,\ldots,X_k$ are independent random variables on $B(\cB)$
for $\cB$ a Banach space with $\rho_{\cB}(\tau) \leq s\tau^2$.  Then
\be \Pr\L[ \L \| \frac 1 k \sum_{i=1}^k X_i\R\| \geq \delta \R]
\leq e^{s+2 - c k \delta^2}.\ee
\end{lem}

\begin{proof}[Proof of \lemref{sparse-general}]
First we introduce notation. For a matrix $X$, define the map $\hat
X$ by $\hat X(A) := \langle X,A\rangle$.  Thus $\Lambda = \sum_{i=1}^n
Y_i \hat X_i$.
 
As in \lemref{sparse-basic} we first drop terms with $Y_i=0$ and
rewrite $\Lambda = \sum_{i=1}^n \frac{Y_i}{\|Y_i\|_{\cB}}\cdot
\|Y_i\|_{\cB} \hat X_i$.   Redefine $X_i,Y_i$ appropriately and assume from
now on that each $\|Y_i\|_{\cB} =  1$.

Define $p_i = \tr[X_i] / d$.  Note that $p\in \bbR^n_+$ and
$\|p\|_1\leq 1$.  Sample $i_1,\ldots,i_k$ according to $p$ and let
them take 
value 0 with probability $1-\sum_i p_i$.  Set $ \Lambda'' =
\sum_{j=1}^k Y_j' \hat X_j'$ where $Y_j' = Y_{i_j}$ and $X_j' =
\frac{X_{i_j}}{kp_{i_j}}$.  (Set $X_j' = Y_j' = 0$ if $i_j=0$.)
 These choices mean that $\E[\Lambda''] = \Lambda$.

Let $\bar X := \sum_{i=1}^n X_i$ and observe that $0 \leq \bar X \leq
I$.  Additionally $\E[X_j'] = \bar X / k$.  Thus if we define $Z_j :=
kX_j' - \bar X $ then $\E[Z_j] = 0$ and $\|Z_j\| \leq d$.  The
operator Hoeffding bound (\lemref{Hoeff}) implies that $\|\frac 1 k \sum_{j=1}^k
Z_j \| \leq \delta$ with probability $\geq 1 - d
\exp(-k\delta^2/8d^2) $.   When this occurs we have $\|\sum_{j=1}^k
X_j' - \bar X\| \leq \delta$ and thus 
\be \sum_{j=1}^k X_j' \leq (1+\delta)I.
\label{eq:nearly-POVM}\ee

Next we attempt to bound the LHS of Eq. \eq{SB-approx}.   First we can
relax $\cD_d$ to $B(S_1)$ and obtain
\be \max_{\rho \in \cD_d} \| (\Lambda'' - \Lambda)(\rho)\|_{\cB}
\leq \| \Lambda'' - \Lambda \|_{S_1\ra \cB}
= \| \Lambda'' - \E[\Lambda''] \|_{S_1\ra \cB}.\ee
This formulation allows to apply the symmetrization trick
(Lemma \ref{symmetrization}) to obtain
\be \E_{i_1,\ldots,i_k\sim p^{\ot k}}
\max_{\rho \in \cD_d} \| (\Lambda'' - \Lambda)(\rho)\|_{\cB}
\leq
2  \E_{i_1,\ldots,i_k\sim p^{\ot k}} \E_{\eps_1,\ldots,\eps_k}
\L \| \frac 1 k \sum_{j=1}^k \eps_j Y_{i_j} \frac{\hat X_{i_j}}{p_{i_j}}\R\|_{S_1\ra
  \cB} \label{eq:sym-Lambda}\ee
We will bound this last quantity for any fixed $i_1,\ldots,i_k$.
For $\rho\in B(S_1)$, %(i.e. $\|\rho\|_1\leq 1$ but $\rho$ is not
%  necessarily psd)
 define $q_j := \langle X_{i_j},\rho\rangle / kp_{i_j}$.  Denote the set of feasible $q$ by $S_{X,\vec i}$ where this notation emphasizes the dependence on both $X$ and $i_1,\ldots,i_k$.  Then $\sum_j |q_j| \leq 1$, each $|q_j| \leq d/k$ and
\be \L \| \underbrace{\frac 1 k \sum_{j=1}^k \eps_j Y_{i_j} \frac{\hat
    X_{i_j}}{p_{i_j}}}_{=:\Lambda'_\eps}
\R\|_{S_1\ra
  \cB} = \max_{q \in S_{X, \vec i}}\L \|\sum_{j=1}^k \eps_j Y_{i_j}
q_j \R\|_{\cB}.
\label{eq:fixed-i-norm}\ee

Now let us fix $\rho$ (or equivalently $q$).   Observe that
$\|Y_{i_j} q_j\|_{\cB}\leq d/k$.  Then \lemref{azuma} implies that
\be \Pr_{\eps_1,\ldots,\eps_k}\L[\L \|\sum_{j=1}^k \eps_j Y_{i_j} q_j
\R\|_{\cB} \geq \delta \R]
\leq e^{s + 2 - \frac{ck\delta^2}{d^2}}.
\label{eq:eps-good}\ee
According to Lemma II.2 of \cite{HLSW04} there exists a net of pure
states $\rho_1,\ldots,\rho_m \in \cD_d$ such that $m\leq 10^{2d}$ and for any pure state
$\rho$, we have $\min_l \|\rho - \rho_l\|_1\leq 1/2$.  Say that
$\eps_1,\ldots,\eps_k$ is a good sequence if $\|\sum_j \eps_j Y_{i_j}
\langle X_{i_j}, \rho_l\rangle / k p_{i_j}\| \leq \delta$ for all
$l\in [m]$.  By Eq. \eq{eps-good} and the union bound the probability that
$\eps_1,\ldots,\eps_k$ is bad (i.e. not good) is $\leq 10^{2d} e^{s + 2 -
  ck\delta^2/d^2}$.   For a bad sequence we still have that
Eq. \eq{fixed-i-norm} is $\leq d$ by the triangle inequality.  For a good
sequence, let $\alpha$ denote Eq. \eq{fixed-i-norm} and let $\beta$ be the
corresponding maximum with $\rho$ restricted to the set
$\{\rho_1,\ldots,\rho_m\}$.  By our assumption that the sequence is
good we have $\beta \leq \delta$.  Observe that
$\alpha = \max_{\rho \in B(S_1)} \|\Lambda'_\eps(\rho)\|_{\cB}$ and by
convexity (and symmetry of the $\|\cdot\|_{\cB}$ norm) this $\max$ is
achieved for $\rho$ a pure state.  Let $\rho_l$ satisfy $\|\rho -
\rho_l\|_1\leq 1/2$.  Then
\be \|\Lambda'_\eps(\rho)\|_{\cB} 
\leq \|\Lambda'_\eps(\rho_l)\|_{\cB}  
+ \|\Lambda'_\eps(\rho-\rho_l)\|_{\cB} 
\leq \beta
+ \alpha \cdot \frac 1 2.\ee

Maximizing the LHS over $\rho$ we obtain $\alpha \leq \beta +
\alpha/2$, or equivalently $\alpha \leq 2\beta\leq 2\delta$.
Thus Eq. \eq{sym-Lambda} is
\be \leq 4 \delta + 2 d 10^{2d} e^{s + 2 - \frac{ck\delta^2}{d^2}}.\ee
Redefining $c$, this is $\leq 5\delta$ when  $k \geq cd^3/\delta^2$.

Since Eq. \eq{sym-Lambda} controls the expectation with respect to
$i_1,\ldots,i_k$, we conclude that for at least half of the
$i_1,\ldots,i_k$, the LHS of Eq. \eq{SB-approx} is $\leq 10\delta$.
Since Eq. \eq{nearly-POVM} holds with high probability ($\geq 1-
d\exp(-cd/8)$) it follows that there exists a sequence of
$i_1,\ldots,i_k$ that simultaneously fulfills both criteria. 
Fix this choice. Finally we choose $\Lambda' = \Lambda''/ (1+\delta)$ so that the normalization
condition on $\sum_j X_j'$ is satisfied. This increases the error by at
most a further factor of $\delta$. We conclude the proof by
redefining $\delta$ to be $11 \delta$.
\end{proof}

%\section{Multipartite algorithm}\label{sec:multi}

\section{Hardness of computing $2 \rightarrow q$ norms} \label{hardness}

In this section we extend the hardness results of \cite{BHKSZ12} (Theorem 9.4, part 2) for estimating the $2 \rightarrow 4$ norm to general $2 \rightarrow q$ norms for even $q\geq 4$. 

The next lemma is an extension from Lemma 9.5 from  \cite{BHKSZ12}.

\begin{lem} \label{lem95aux}
Let $M \in L(\mathbb{C}^d \otimes \mathbb{C}^d)$ satisfy $0 \leq M \leq I$. Assume that either (case $Y$) $h_{\text{Sep}(d,d)}(M) = 1$ or (case $N$) $h_{\text{Sep}(d,d)}(M) \leq 1 - \delta$. Let $k$ be a positive integer and $q \geq 4$ an even positive integer. Then there exists a matrix $A$ of size $d^{4kq} \times d^{2kq}$ such that in case $Y$, $\Vert A \Vert_{2 \rightarrow q} = 1$, and in case $N$, $\Vert A \Vert_{2 \rightarrow q} \leq (1 - \delta/2)^k$. Moreover, $A$ can be constructed efficiently from $M$. 
\end{lem}

\begin{proof}
Consider the following operator
\begin{equation}
N := (M_{A_1B_1}^{1/2} \otimes \ldots \otimes M_{ A_{q/2} B_{q/2} }^{1/2}) P_{A_1, \ldots, A_{q/2}} \otimes P_{B_1, \ldots, B_{q/2}} (M_{A_1B_1}^{1/2} \otimes \ldots \otimes M_{ A_{q/2} B_{q/2} }^{1/2}) ,
\end{equation}
with $P_{A_1, \ldots, A_{q/2}}$ the projector onto the symmetric subspace over $A_1, ..., A_{q/2}$. We will first relate $h_{ \text{Sep}^{q/2}(d^2) } (N)$ to $h_{\text{Sep}(d,d)}(M)$, and then relate $h_{\text{Sep}^{q/2} (d^2) } (N)$ to $\Vert A \Vert_{2 \rightarrow q}$ for a matrix $A$ of size $d^{4kq} \times d^{2kq}$.

First we show that in case $Y$, $h_{\Sep^{q/2}(d^2)}(N) =1$. Indeed since there are unit vectors $x, y \in \mathbb{C}^d$ satisfying $M_{AB} (x \otimes y) = x \otimes y$, we have
\begin{eqnarray}
h_{\text{Sep}^{q/2}(d^2)}(N) &=& \max_{v_1, \ldots, v_{q/2} \in \mathbb{C}^{d^2}}   (v_1 \otimes \ldots \otimes v_{q/2})^*   N  (v_1 \otimes \ldots \otimes v_{q/2}) \nonumber \\ 
&\geq&    (x^{\otimes q/2} \otimes y^{\otimes q/2})^*   N  (x^{\otimes q/2} \otimes y^{\otimes q/2})     \nonumber \\ 
&=&  (x^{\otimes q/2} \otimes y^{\otimes q/2})^*  P_{A_1, \ldots, A_{q/2}} \otimes P_{B_1, \ldots, B_{q/2}}  (x^{\otimes q/2} \otimes y^{\otimes q/2})  = 1 
\nonumber 
\end{eqnarray}

In case $N$ we show that $h_{\Sep^{q/2}(d^2)}(N) \leq 1 - \delta/2$. Note that 
\be P_{A_1, \ldots, A_{q/2}} \leq P_{A_1A_2} \otimes I_{A_3 \ldots A_{q/2}}
\label{eq:wasteful}.\ee
Then  
\begin{eqnarray}
h_{\Sep^{q/2}(d^2)}(N) &=&  \max_{v_1, \ldots, v_{q/2} \in \mathbb{C}^{d^2}}   (v_1 \otimes \ldots \otimes v_{q/2})^*  N  (v_1 \otimes \ldots \otimes v_{q/2}) \nonumber \\ 
&\leq&   \max_{v_1, v_{2} \in \mathbb{C}^{d^2}}   (v_1 \otimes v_{2})^*  (M_{A_1B_1}^{1/2} \otimes M_{A_2B_2}^{1/2})P_{A_1A_2} \otimes P_{B_1B_2} (M_{A_1B_1}^{1/2} \otimes M_{A_2B_2}^{1/2})   (v_1 \otimes v_{2})  \nonumber \\ 
&\leq& 1-\delta/2,
\end{eqnarray}
where the last inequality follows from Lemma 9.6 of  \cite{BHKSZ12}. 

To construct a matrix $A$ of size $d^{4kq} \times d^{2kq}$ s.t. $\Vert A \Vert_{2 \rightarrow q} = h_{\Sep^{q/2}(d^2)}(N)$ we follow the proof of Lemma 9.5 of  \cite{BHKSZ12}, the only difference being that we apply Wick's theorem to $P_{A_1, \ldots A_{q/2}}$, i.e. there is a measure $\mu$ over unit vectors s.t. 
\be
P_{A_1, \ldots, A_{q/2}} =  \binom{d + q/2 -1}{q/2} \int \mu(dv) (vv^*)^{\otimes q/2}.
\ee
\end{proof}

The basic idea of the Lemma is to use the product test of \cite{HM10} to force $v_1,\ldots,v_{q/2}$ to be product states. Our proof can be summarized as saying that $q/2$ copies can enforce this more effectively than 2 copies (assuming $q/2\geq 2$), and therefore we obtain soundness at least as sharp as in \cite{BHKSZ12}.  This analysis may be wasteful, since using more copies should {\em improve} the effectiveness of the product test.

The main result of this section is the following analogue of Theorem 9.4, part 2, of \cite{BHKSZ12}:

\begin{thm}
Let $\phi$ be a $3$-SAT instance with $n$ variables and $O(n)$ clauses and $q \geq 4$ an even integer. Determining whether $\phi$ is satisfiable can be reduced in polynomial time to determining whether $\Vert A \Vert_{2 \rightarrow q} \geq C$ or $\Vert A \Vert_{2 \rightarrow q} \leq c$ where $0 \leq c < C$ and $A$ is an $m \times m$ matrix, where $m = \exp(q \sqrt{n} \polylog(n) \log(C/c))$.  
\end{thm}

This gives nontrivial hardness for super-constant $q$, in fact up to $\tilde O(\sqrt{\log d})$, but not yet all the way up to $O(\log d)$, where multiplicative approximations are known to be easy.

\begin{proof}
Corollary 14  of [HM10] gives a reduction from determining satisfiability of $\phi$ to distinguishing between $h_{\text{Sep}(d,d)}(M) = 1$ and $h_{\text{Sep}(d,d)}(M) \leq 1/2$, with $0 \leq M \leq I$ that can be constructed in time $\poly(d)$ from $\phi$ with $d = \exp(\sqrt{n} \polylog(n))$. Applying Lemma \ref{lem95aux} gives the result.
\end{proof}

\bibliographystyle{hyperalpha}
\bibliography{refs}
\end{document}